\documentclass[10pt,reqno,final]{article}

\usepackage{amsmath,amsfonts,amssymb,amsthm,version}
\usepackage{mathrsfs,fancybox,pifont}
\usepackage{graphicx}
\usepackage{url,hyperref}
\usepackage[notcite,notref]{showkeys}
\usepackage{color}
\usepackage{subfigure,multirow}
\usepackage{epstopdf}
\usepackage{cases}
\usepackage{mathtools}
\usepackage{algorithm,algorithmic}
\usepackage{authblk}
\usepackage{fancyhdr}
\usepackage{lipsum}
\usepackage{cite}
\usepackage{float}

\allowdisplaybreaks

\numberwithin{equation}{section}
\numberwithin{figure}{section}
\numberwithin{table}{section}

\theoremstyle{plain}
\newtheorem{theorem}{Theorem}[section]
\newtheorem{lemma}{Lemma}[section]

\theoremstyle{definition}
\newtheorem{definition}{Definition}[section]

\theoremstyle{remark}

\title{Dynamics of infectious diseases in predator-prey populations: a stochastic model, sustainability, and invariant measure}
\author[1]{Yujie Gao} 
\author[2]{Malay Banerjee}
\author[1]{Ton Viet Ta}
\affil[1]{Mathematical Modeling Laboratory, Kyushu University}
\affil[2]{Department of Mathematics and Statistics, IIT Kanpur}
\date{\today}

\begin{document}
\maketitle

\begin{abstract}
This paper introduces an innovative model for infectious diseases in predator-prey populations. We not only prove the existence of global non-negative solutions but also establish essential criteria for the system's decline and sustainability. Furthermore, we demonstrate the presence of a Borel invariant measure, adding a new dimension to our understanding of the system. To illustrate the practical implications of our findings, we present numerical results. With our model's comprehensive approach, we aim to provide valuable insights into the dynamics of infectious diseases and their impact on predator-prey populations.

  \medskip
  \noindent{\bf Keywords}: Holling-Tanner model, infectious diseases, stochastic differential equations 

  \medskip
\end{abstract}

\section{Introduction}

The predator-prey relationship is one of the fundamental concepts in mathematical biology, and several models have been proposed to describe this relationship. One such model, originally introduced by Leslie \cite{leslie1948some}, has attracted considerable attention from researchers. This model takes the form of a system of differential equations:
\begin{equation}\label{equ01}
\left\{\begin{aligned}
&\frac{d x(t)}{d t}=x(t)[a_1-b_1x(t)]-p(x)y(t), \\
&\frac{d y(t)}{d t}=y(t)\left[a_2-b_2\frac{y(t)}{x(t)}\right].
\end{aligned}\right.
\end{equation}
Here, $x(t)$ and $y(t)$ denote the densities of the prey and predator populations at time $t$, respectively. The function $p(x)$ describes the predator's response to the prey population.  The prey and predator populations grow logistically with intrinsic growth rates of $a_1$ and $a_2$, respectively, and with environmental carrying capacities of $\frac{a_1}{b_1}$ and $\frac{a_2 x}{b_2}$.
Here $p(x)=\frac{ax}{b+x}$, is known as the Michaelis–Menten type saturating functional response \cite{michaelis1913kinetik}. Furthermore, \eqref{equ01}  is referred to as a Holling-Tanner predator-prey system \cite{zuo2016periodic}:
\begin{equation}\label{mulequa1}
  \left\{\begin{aligned}
  &\frac{d x(t)}{d t}=x(t)[a_1-b_1x(t)]-\frac{ax(t)y(t)}{b+x(t)}, \\
  &\frac{d y(t)}{d t}=y(t)\left[a_2-b_2\frac{y(t)}{x(t)}\right].
  \end{aligned}\right.
  \end{equation}

When $p(x)=\frac{ax}{b+y+x} $, it is known as the Beddington–DeAngelis functional response. 
In \cite{mandal2012stochastic}, authors considered a deterministic predator–prey model with Beddington–DeAngelis functional response and generalist predators
\begin{equation}\label{mulequa2}
  \left\{\begin{aligned}
  &\frac{d x(t)}{d t}=x(t)[a_1-b_1x(t)]-\frac{ax(t)y(t)}{b+y(t)+x(t)}, \\
  &\frac{d y(t)}{d t}=y(t)\left[a_2-b_2\frac{y(t)}{c+x(t)}\right].
  \end{aligned}\right.
  \end{equation}
and its stochastic version. Here, the per-capita growth rate of generalist predator population is modified to $a_2-b_2\frac{y(t)}{c+x(t)}$ in \eqref{mulequa2} based upon the assumption that the effective carrying capacity is proportional to the sum of favourite prey population and alternative implicit food source. Effective environmental carrying capacity of generalist predators is equal to $\frac{a_2(c+x)}{b_2}$. 

It is worth noting that the environmental carrying capacities of predators in equations \eqref{equ01} and \eqref{mulequa1}, denoted by $\frac{a_2 x}{b_2}$ and in equation \eqref{mulequa2}, denoted by $\frac{a_2(c+x)}{b_2}$, have lower limits of $0$ and $\frac{a_2c}{b_2}$, respectively. However, these capacities do not have an upper bound that is independent of the initial values. Therefore, in cases where the initial population of prey is very large,  the carrying capacity of generalist predators can also be very high.

In the real world, infectious diseases are a common occurrence in populations. One of the most well-known models for infectious diseases is the SIR model, which was introduced by Kernmack and McKendrick \cite{kernackmckendrick}. This model is described by a system of differential equations,:
 \begin{equation}\label{mulequa4}
  \left\{\begin{aligned}
  &\frac{d S(t)}{d t}=-\frac{\beta }{N}S(t)I(t), \\
  &\frac{d I(t)}{d t}=\frac{\beta }{N}S(t)I(t)-\gamma I(t), \\
  &\frac{d R(t)}{d t}=\gamma I(t).
  \end{aligned}\right.
  \end{equation}
The population is divided into three compartments:: ``susceptible" $S(t)$, ``infected`` $I(t)$, and ``recovered`` $R(t)$, where $N=S(t)+I(t)+R(t)$ remains constant over time. The function $S(t)$ represents the number of individuals who have not been infected with the disease or who are still susceptible to it. The function $I(t)$ represents the number of infected individuals who can infect the susceptible individuals. The function $R(t)$ denotes the number of individuals who have recovered from the disease, including those who have died and can no longer be infected or transmit the disease to others.

The parameter $\beta > 0$ represents the average number of individuals that an infected person can infect per day, while $\gamma > 0$ denotes the recovery rate, $\frac{1}{\gamma}$ measures the average recovery time. The SIR model assumes a short-lived outbreak and does not consider natural birth and death rates. Moreover, it assumes that recovered individuals acquire lifetime immunity and are immune to the disease.

In real-world ecological systems, infectious diseases can impact both predator and prey populations. Several researchers have investigated the dynamics of infectious diseases in predator-prey models. For instance, Pierre et al. \cite{prtgmj} developed a deterministic model in which the predator species is affected by a disease, while Xiao and Chen \cite{XIAO200159} analyzed a predator-prey model with disease in the prey population. Eilersen et al. \cite{EilersenJensen} investigated a deterministic model with one predator and two prey species, in which one prey species carries a disease. In contrast, Li and Wang \cite{liwang} studied a classical stochastic predator-prey model with disease in the predator population. Notably, their models assume that the predator species declines if there is no prey, i.e., the  predators are specialist and they have no alternative food source.

In this paper, we investigate a predator-prey model that includes an infectious disease in the predator species and exhibits logistic growth. We assume that the environmental carrying capacity of predators is bounded both above and below by constants. Unlike the SIR model \eqref{mulequa4} introduced earlier, we assume that any recovered predator individual can be infected again, and the infection rate is the same as that for susceptible individuals.

Let $x(t)$ denote the prey density, $S(t)$ denote the density of susceptible predators, and $I(t)$ denote the density of infected predators at time $t$. Our model is described by the following system of stochastic differential equations:

\begin{equation}\label{mulequa5}
\begin{cases}
  \begin{aligned}
  d x(t)=&\left[ x(t)\{a_1-b_1x(t)\}-\frac{ax(t)(S(t)+I(t))}{b+S(t)+I(t)+x(t)} \right]d t+ \sigma_1x(t)dW_1(t), \\
  d S(t)=&\left[S(t)\left(a_2-\left \{ \frac{1}{c+ x(t)}+h\right \}  \{S(t)+I(t)\} \right)-\beta(S,I)S(t)I(t)+\gamma I(t) \right]dt\\
 &+ \sigma_2S(t)dW_2(t), \\
  d I(t)=&[\beta (S,I) S(t)I(t)-\alpha I(t) ]d t+ \sigma_3I(t)dW_3(t)
  \end{aligned}
  \end{cases}
\end{equation}
coupled with non-negative initial conditions $x(0), y(0), z(0).$ 
Here we divide the intra-specific competition term among the predators into two parts: one depends upon the available resources and the other is independent to the available resources. The resource independent intra-specific competition rate is denoted by $h$.

In this model, we don't assume that the density of predator species is constant, which means that the total predator population $N(t) = S(t) + I(t)$ varies with time. As a result, the rate of infection also varies with time, here $N(t) = S(t) + I(t)$ varies with time continuously. To account for this, we introduce a nonlinear incidence rate function $\beta(S,I)$ that depends on the densities of susceptible and infected predators, given by
$$\beta (S,I) =\frac{\beta_0}{1+\gamma N(t)} =\frac{\beta_0}{1+\gamma(S(t)+I(t))},$$
where $\beta_0$ is a positive constant and  $\gamma>0$ is the recovery rate of infected predators.
 It means that the rate of infection transmission is not directly proportional to the entire size of the infected compartment rather a part of the infected predator can transmit the disease. This is a reasonable assumption as the prey and predator are assumed to be distributed uniformly.
 Therefore, 
$$\beta (S,I) =\frac{\beta^*}{S(t)+I(t)+k},$$
with $k\,=\,\frac{1}{\gamma}$ and $\beta^*\,=\,\beta_0/\gamma$. 


The constant $\alpha$ is the sum of the recovery rate $\gamma$ and the mortality 
$\alpha^*$ of infected predators. Here, $\alpha^* > 0$ represents the probability of a predator dying after being infected. Other parameters have the same meanings as in the models above.

Note that the natural growth rate and mortality of the population are subject to environmental factors that are often random and unpredictable. These factors can be represented by white noise \cite{zuo2016periodic,ta2015sustainability}. Therefore, in equation \eqref{mulequa5}, the parameters $a_1, a_2$, and $\alpha^*$ are influenced by white noise, which can be modeled as a generalized derivative of Brownian motion \cite{cai2016stochastic,rajasekar2020ergodic}:
\begin{equation}\label{mulequa9}
   a_1 \rightarrow a_1 + \sigma_1dW_1(t),  \ a_2 \rightarrow a_2 + \sigma_2dW_2(t), \ \alpha^* \rightarrow \alpha^* -\sigma_3dW_3(t), 
  \end{equation}
where $W_i(t), i=1,2,3$ are independent Brownian motions defined on a filtered complete probability space $(\Omega,\mathcal{F},\{\mathcal F\}_{t \geq 0},\mathbb P)$. Here, the positive constants $\sigma_i$ represent the intensity of the noise.

The present paper is organized as follows: In the next section, we prove the existence of global positive solutions for the stochastic model. Section 3 provides sufficient conditions for the decline of each predator and prey species. We give a sustainability condition for this system in Section 4. In Section 5, we prove the existence of a Borel invariant measure. Section 6  presents numerical illustrations to validate the analytical findings. Finally, the paper concludes with a summary of our key findings in Section 7.

\section{Global solutions}
\label{sec:global}

In this section, we prove the existence of unique global non-negative solutions to \eqref{mulequa5} and provide some properties of the solutions. To establish existence, we utilize the following lemma.


\begin{theorem}\label{thm:mvt}
Let $(x_0,S_0,I_0)\in \overline{\mathbb{R}_+^3}$. Then, there exists a unique global solution $(x(t),S(t),I(t))$ of \eqref{mulequa5}, such that 
$(x(0),S(0),I(0))=(x_0,S_0,I_0)$ and \ $(x(t),S(t),I(t))\in \overline{\mathbb{R}_+^3}$  \ a.s. 
Furthermore, if $(x_0,S_0,I_0)>0$, then $(x(t),S(t),I(t))\in \mathbb{R}_+^3$ a.s.
\end{theorem}

\begin{proof}

Since all functions on the right-hand side of \eqref{mulequa5} are locally Lipschitz
continuous on $\overline{\mathbb{R}_+^3}$, there is a unique local solution $(u_t,v_t)$ defined on an interval $[0,\tau)$, where $\tau $ is a stopping time.
We know that (\cite{mao2007stochastic}), if $\mathbb{P} \{ \tau <\infty\} >0$ then $\tau$ is an explosion time
on ${\tau <\infty}$, i.e.   
$\lim_{t\to\tau}(| x(t)\vert+ | S(t)\vert+ | I(t)\vert )=\infty$ or $\lim_{t\to\tau}[S(t)+I(t)+x(t)+b]=0$  a.s.on ${\tau <\infty}$. 

To establish the existence of global solutions, we must demonstrate that $\tau = \infty$ a.s. To do so, we will examine five cases.

\textbf{Case 1}: $x_0=S_0=I_0=0$. This is a trivial case because the initial densities of all three populations are zero, which implies that they will remain zero for all time. Thus, we have $x(t)=S(t)=I(t)=0$ a.s. for $0<t<\infty$.

\textbf{Case 2}: $x_0>0,S_0>0,I_0=0$. By uniqueness, we have $I(t)=0$ a.s. for all $t\in [0, \tau)$. Therefore, the system of the first two equations in \eqref{mulequa5} reduces to a stochastic Holling-Tanner type model. The existence of global positive solutions $(x(t), S(t))$ can be found in \cite{mandal2012stochastic}.

\textbf{Case 3}: $x_0=0,S_0>0,I_0>0$.

By uniqueness, $x(t)=0$ a.s. for all $t \in [0, \tau)$. Therefore, the system of the two first equations in \eqref{mulequa5} is a SIR model with logistic growth. Similarly to Case 2, the existence of global positive (S(t), I(t)) is given in \cite{mao2007stochastic}. 

\textbf{Case 4}: $x_0>0,S_0>0,I_0>0$. Let $k_0 >$0 be a positive integer such that $x_0$, $S_0$, $I_0$ lie in the interval $[ \frac{1}{k_0}, k_0] $. 

  Denote  
  \begin{displaymath}
    H_k=\left[\frac{1}{k},k\right]\times \left[\frac{1}{k},k\right]\times \left[\frac{1}{k},k\right], \qquad k=1,2,\cdots
  \end{displaymath}
  then $\cup _{k=k_0}^\infty=\mathbb{R}_+^3$. We define a sequence $\{\tau_k\}_{ k=k_0}^\infty$ of stopping times by 
  \begin{displaymath}
    \tau_k= \inf\{ 0<t<\tau; (x(t),S(t),I(t))\notin H_k\}, 
  \end{displaymath}
  with the convention $\inf\varnothing =\infty$. It's clear that this sequence is non-decreasing. Therefore,
  there exists a limit of this sequence, called $\tau_\infty$, and 
  \begin{displaymath}
    \tau_\infty=\lim_{k\to\infty}\tau_k\leqslant \tau  \qquad a.s.  
  \end{displaymath}

  Thus, to prove that $\tau= \infty$, we prove that  $\tau_\infty= \infty$ a.s. Indeed, suppose on the contrary that $\tau_\infty< \infty$.
  Then, there exists $T>0$ and $0< \varepsilon <1$ such that 
  \begin{displaymath}
    \mathbb{P} \{\tau_\infty <T\}>\varepsilon.
  \end{displaymath}

  Consider a positive increasing continuous function $V$ defined by
   \begin{displaymath}
    \begin{aligned}
    V(x,S,I)&=x^2+S^2+I^2-\log x-\log S-\log I \\
    & =V_1+V_2+V_3,       \quad x>0, S>0, I>0,
 \end{aligned}
  \end{displaymath}
 where
  \begin{equation} \label{E1}
    V_1=x^2-\log x, V_2=S^2-\log S,  V_3=I^2-\log I.
  \end{equation}
  Then,
  \begin{displaymath}
    dV= dV_1+dV_2+dV_3.
  \end{displaymath}
  Taking It$\hat{o}$ formula into $V_1$ first, we have 
  \begin{displaymath}
    dV_1 = \left[ \mu \left(2x-\frac{1}{x}\right)+\frac{1}{2}\sigma_1^2x^2\left(2+\frac{1}{x^2}\right ) \right] dt + \sigma_1x\left(2x-\frac{1}{x} \right)dW_1,
  \end{displaymath} 
  where $\mu=x(t)(a_1-b_1x(t))-\frac{ax(t)(S(t)+I(t))}{b+S(t)+I(t)+x(t)}$.
  Thus,
    \begin{align}
    dV_1 =& \big[[x(t)(a_1-b_1x(t))-\frac{ax(t)(S(t)+I(t))}{b+S(t)+I(t)+x(t)}] (2x(t)-\frac{1}{x(t)}) \notag \\
    & +\frac{1}{2}\sigma_1^2x^2(t)(2+\frac{1}{x^2(t)} ) \big] dt + \sigma_1x(t)(2x(t)-\frac{1}{x(t)} )dW_1  \notag \\
     =&\{2x^2(t)(a_1-b_1x(t))-(a_1-b_1x(t))-\frac{2ax^2(t)(S(t)+I(t))}{b+S(t)+I(t)+x(t)}+\frac{a(S(t)+I(t))}{b+S(t)+I(t)+x(t)}   \notag \\
    &+\sigma_1^2x^2(t)+\frac{1}{2}\sigma_1^2\}dt + \sigma_1x(t)(2x(t)-\frac{1}{x(t)} )dW_1   \label{E2} \\
    \leqslant  & \{2a_1x^2(t)-(a_1-b_1x(t))+\frac{a(S(t)+I(t))}{b+S(t)+I(t)+x(t)}+\sigma_1^2x^2(t)+\frac{1}{2}\sigma_1^2\}dt \notag \\
    &+ \sigma_1x(t)(2x(t)-\frac{1}{x(t)} )dW_1. \notag
  \end{align}

Since $x(t)$, $S(t)$, $I(t)$ are positive for $t \in [0, \tau_\infty)$, there exist $M_1, M_2$ such that
\begin{displaymath}
  \begin{aligned}
 & 2a_1x^2(x)-(a_1-b_1x(0))+\frac{a(S(t)+I(t))}{b+S(t)+I(t)+x(t)}+\sigma_1^2x^2(t)+\frac{1}{2}\sigma_2^2  \\
 &\leqslant 2a_1x^2(t) -(a_1-b_1x(t))+a+\sigma_1^2x^2(t)+\frac{1}{2}\sigma_2^2  \\
 &\leqslant M_1V_1+M_2. 
  \end{aligned}
\end{displaymath}
Thus,
\begin{equation}\label{eeq1}
  \begin{aligned}
  dV_1(t) &\leqslant  (M_1V_1+M_2)dt +\sigma_1x(2x-\frac{1}{x} )dW_1,   \qquad t\in [0, \tau_\infty).
\end{aligned}
\end{equation}
Taking the same argument to $V_2, V_3$, we have
\begin{displaymath}
  \begin{aligned}
  dV_2 &= [ [ S(t)(a_2-(\frac{1}{x(t)+e}+h)(S(t)+I(t)))-\beta (S,I) I(t)S(t)+\gamma I(t)] (2S(t)-\frac{1}{S(t)})  & \\
  & \quad+\frac{1}{2}\sigma_2^2S^2(t)(2+\frac{1}{S^2(t)})] d t+(2\sigma_2S^2(t)-\sigma_2)d W_2 & \\
  & = [2S^2(t)\{a_2-(\frac{1}{x(t)+e}+h)(S(t)+I(t))\}-2\beta (S,I)I(t)S^2(t)+2\gamma I(t)S(t) & \\
  & \quad-\{a_2-(\frac{1}{x(t)+e}+h)(S(t)+I(t))\} + \beta(S,I)I(t)-\gamma \frac{I(t)}{S(t)}+\sigma_2^2S^2(t)+\frac{1}{2}\sigma_2^2] d t & \\
  & \quad+ (2\sigma_2S^2(t)-\sigma_2)d W_2 & \\
  &\leqslant [2a_2S^2(t)+2\gamma I(t)S(t)+(\frac{1}{x(t)+e}+h)(S(t)+I(t))+ \beta(S,I)I(t)+\sigma_2^2S^2(t)+\frac{1}{2}\sigma_2^2] d t & \\
  & \quad+(2\sigma_2S^2(t)-\sigma_2)d W_2,& \\
\end{aligned}
\end{displaymath}
and
\begin{equation}  \label{E3}
\begin{aligned}
  dV_3 =& [(\beta(S,I)S(t)I(t)- \alpha I(t))(2I(t)-\frac{1}{I(t)})+\sigma_3^2I^2(t)+\frac{1}{2}\sigma_3^2] d t \\
& +(2\sigma_3I^2(t)-\sigma_3)d W_3   \\
  =& [2\beta(S,I)S(t)I^2(t)-2(\alpha-\gamma)  I^2(t)-2\gamma I^2(t)-\beta(S,I)S(t)+\alpha+\sigma_3^2I^2(t)+\frac{1}{2}\sigma_3^2] d t  \\
  &+(2\sigma_3I^2(t)-\sigma_3)d W_3.   
\end{aligned}
\end{equation}

Since $\beta(S,I) =\frac{\beta^*}{N+k}  $ where k is a constant, 
$\beta(S,I)S(t)\leqslant \beta^*,$ and $\beta(S,I)I(t) \leqslant \beta^*$. 
There exist $M_3, M_4>0$ such that
\begin{displaymath}
  \begin{aligned}
   &2a_2S^2(t)+2\gamma I(t)S(t)+(\frac{1}{x(t)+e}+h)(S(t)+I(t))+ \beta(S,I)I(t)+\sigma_2^2S^2(t)+\frac{1}{2}\sigma_2^2 d t &\\
   & \leqslant M_3V_2+\frac{1}{2}M_3V_3+M_4, &\\
   \end{aligned}
\end{displaymath}
and
\begin{displaymath}
  \begin{aligned}
  &2\beta(S,I)S(t)I^2(t)-2(\alpha-\gamma)  I^2(t)-2\gamma I^2(t)-\beta(S,I)S(t)+\alpha+\sigma_3^2I^2(t)+\frac{1}{2}\sigma_3^2 d t \leqslant \frac{1}{2}M_3V_3+M_4.
   \end{aligned}
\end{displaymath}
As a consequence,
\begin{equation}\label{eeq2}
  \begin{aligned}
 dV_2 \leqslant M_3V_2+\frac{1}{2}M_3V_3+M_4) dt+[2\sigma_2S^2(t)-\sigma_2]d W_2(t),
   \end{aligned}
\end{equation}
and
\begin{equation}\label{eeq3}
  \begin{aligned}
 dV_3 \leqslant (\frac{1}{2}M_3V_3+M_4) dt+[2\sigma_3I^2(t)-\sigma_3]d W_3(t).
   \end{aligned}
\end{equation}
From \eqref{eeq1}, \eqref{eeq2} and \eqref{eeq3}, we have
\begin{displaymath}
  \begin{aligned}
  dV &\leqslant (K_1V+K_2)dt+\sigma_1x(t)(2x(t)-\frac{1}{x(t)} )dW_1(t) + (2\sigma_2S^2(t)-\sigma_2)d W_2(t) &\\ 
  & \quad+(2\sigma_3I^2(t)-\sigma_3)d W_3(t), &\\
  \end{aligned}
\end{displaymath}
where $K_1=3\max\{M_1,M_3\}$, $K_2=3\max\{M_2,M_4\}.$

Therefore, we have 
\begin{displaymath}
  \begin{aligned}
  \int_{0}^{t\wedge \tau_k}  \,dV &\leqslant \int_{0}^{t\wedge \tau_k} (K_1V+K_2) \,ds+ \int_{0}^{t\wedge \tau_k} \sigma_1x(s)(2x(s)-\frac{1}{x(s)} )\,dW_1  &  \\
  & + \int_{0}^{t\wedge \tau_k} (2\sigma_2S^2(s)-\sigma_2)\,dW_2 + \int_{0}^{t\wedge \tau_k} (2\sigma_3I^2(s)-\sigma_3)\,dW_3,    \qquad       t\in [0,T].
  \end{aligned}
\end{displaymath}

Taking expectation in the two sides of the above inequality, we obtain
\begin{displaymath}
    \begin{aligned}
    \mathbb{E} \int_{0}^{t\wedge \tau_k}  \,dV &\leqslant \mathbb{E} \int_{0}^{t\wedge \tau_k} (K_1V+K_2) \,dt+\mathbb{E}  \int_{0}^{t\wedge \tau_k} \sigma_1x(s)(2x(s)-\frac{1}{x(s)} )\,dW_1&  \\
    & + \mathbb{E} \int_{0}^{t\wedge \tau_k} (2\sigma_2S^2(s)-\sigma_2)\,dW_2 +\mathbb{E}  \int_{0}^{t\wedge \tau_k} (2\sigma_3I^2(s)-\sigma_3)\,dW_3,   \qquad t \in [0,T].& \\
  \end{aligned}
\end{displaymath}
Then,
  \begin{align*}
  \mathbb{E} [V(x(t\wedge \tau_k), S(t\wedge \tau_k), I(t\wedge \tau_k))-V(x_0, S_0, I_0)] &\leqslant K_2\mathbb{E} (t\wedge \tau_k)+\mathbb{E} \int_{0}^{t\wedge \tau_k} K_1V \,dt  \\
  &\leqslant K_2T+\int_{0}^{t\wedge \tau_k} \mathbb{E} K_1V \,dt, \qquad t \in [0,T].  
  \end{align*}
According to the Gronwall inequality, we can draw a conclusion that
\begin{equation*}
  \begin{aligned}
    \mathbb{E} V(x_{t\wedge \tau_k}, S_{t\wedge \tau_k}, I_{t\wedge \tau_k}) &\leqslant [V(x_0, S_0, I_0)+K_2T]e^{K_1t}, \qquad  t \in [0,T].&  
  \end{aligned}
\end{equation*}
Thus,
\begin{equation}  \label{mulequa11}
  \begin{aligned}
    \mathbb{E} V(x_{T\wedge \tau_k}, S_{T\wedge \tau_k}, I_{T\wedge \tau_k}) &\leqslant [V(x_0, S_0, I_0)+K_2T]e^{K_1T}.    
  \end{aligned}
\end{equation}
On the other hand,
  \begin{align}
    \mathbb{E} V(x(T\wedge \tau_k), S(T\wedge \tau_k), I(T\wedge \tau_k)) &\geqslant \mathbb{E} \textbf{1}_{\{\tau_\infty<T\}}V(x(T\wedge \tau_k), S(T\wedge \tau_k), I(T\wedge \tau_k)) \notag \\
    &\geqslant \mathbb{E} \textbf{1}_{\{\tau_\infty<T\}}V(x(\tau_k), S(\tau_k), I(\tau_k))    \notag \\
    &\geqslant \varepsilon \mathbb{E}V(x(\tau_k), S(\tau_k), I(\tau_k)),  \label{mulequa12}
  \end{align}
(because we have $\mathbb{P} (\tau_\infty<T)>\varepsilon$).

From (\ref{mulequa11}) and (\ref{mulequa12}), we have
\begin{equation}\label{mulequa13}
  \begin{aligned}
  \varepsilon \mathbb{E} V(x(\tau_k), S(\tau_k), I(\tau_k))\leqslant [V(x_0, S_0, I_0)+K_2 T]e^{K_1T} < \infty,  \qquad    t \in [0,T].
  \end{aligned}
\end{equation}
According to the definition of $\tau_k$, we find that $(x(\tau_k), S(\tau_k), I(\tau_k))\in \partial H_k$. Therefore, 
\begin{displaymath}
  \begin{aligned}
\mathbb{E} V(x(\tau_k), S(\tau_k), I(\tau_k))& \geqslant min\{K^2-\log K, (\frac{1}{K})^2-\log \frac{1}{K}\}  & \\
& = min\{K^2-\log K, (\frac{1}{K})^2+\log K\}. 
 \end{aligned}
\end{displaymath}
Thus, taking $K \to \infty$ and using \eqref{mulequa13}, we have $\infty  \leq  [V(x_0, S_0, I_0)+K_2T]e^{K_1T} < \infty$. This contradiction implies that $\tau_\infty=\infty$ a.s. Thus, $\tau = \infty$ a.s.

\textbf{Case 5}: $x_0>0, S_0=0, I_0>0$. Consider two stopping times $\tau^\ast_1=\inf\{0<t<\tau: x(t)=0\}$ and $\tau^\ast_2= \inf\{0<t<\tau: I(t)=0 \}$ with the convention $\inf \varnothing =\infty$. Obviously, we have $x(t)>0$ in $[0,\tau^\ast_1)$ and $I(t)>0$ in $[0,\tau^\ast_2)$. 

Putting $\tau^\ast = \min \{\tau^\ast_1, \tau^\ast_2\}$, we have $x(t)>0$ and $I(t)>0$ in $[0, \tau^\ast)$.


Firstly, we prove that $S(t)>0$ in $(0, \tau^\ast)$. Consider the equations:
$$
    d S(t) = \{S(t)[\textcolor{red}{a2}-(\frac{1}{x(t)+e}+h)(S(t)+I(t))]-\beta(S,I)S(t)I(t)+\gamma I(t)\}d t + \sigma_2S(t)dW_2(t), $$
and
$$
    d \underline{S} (t) = \{\underline{S}(t)[a2-(\frac{1}{x(t)+e}+h)(\underline{S}(t)+I(t))] -\beta(S,I)\underline{S}(t)I(t)\}d t + \sigma_2\underline{S}(t)dW_2(t),  
$$
where $\underline S(0)=S(0)=0.$ 
We have $\underline{S} (t) = 0$. Since $\gamma I(t) > 0$ on $[0, \tau^\ast_2)$, using the comparison theorem, we have $S(t) \geqslant 0$, $t \in [0, \tau^\ast_2)$.

Let us now show that $\mathbb{E} I(t)$ is bounded above. From
  \begin{displaymath}
    \begin{aligned}
      d I(t)&=\{\beta(S,I)S(t)I(t)- \alpha I(t) \}d t+\sigma_3I(t)dW_3(t), & \\
      \end{aligned}
  \end{displaymath}
  we have
  \begin{displaymath}
    \begin{aligned}
      d \ln I(t)&=\{\beta(S,I)S(t)- \alpha  -\frac{1}{2} \sigma_3^2 \}d t+ \sigma_3dW_3(t).  & \\
      \end{aligned}
  \end{displaymath}
Taking the integration of both the hand sides, we have
   $$   \ln \frac{I(t)}{I_0} =\int_{0}^{t} [\beta (s)S(s)- \alpha  -\frac{1}{2} \sigma_3^2 ]ds + \int_{0}^{t} \sigma_3 \,dW_3(s).$$
Then, 
    \begin{align*}
      \mathbb{E} I(t)  &   \leqslant  I_0   \mathbb{E} e^{  \int_{0}^{t} [\beta (s)S(s)-\gamma -(\alpha-\gamma)-\frac{1}{2} \sigma_3^2]ds+ \int_{0}^{t} \sigma_3 dW_3(s)}  \\
      &\leqslant   I_0   \mathbb{E} e^{\int_{0}^{t} [\beta^*-\alpha-\frac{1}{2} \sigma_3^2] ds+ \int_{0}^{t} \sigma_3 dW_3(s)}  \qquad   t\in [0, \bar{T}],
    \end{align*}
  where $\bar{T}$ is a positive constant.
 Since 
 $\mathbb{E}  e^ {\int_{t_0}^{t}  \alpha (s)\,dW_s } =e^{ \frac{1}{2} \int_{t_0}^{t} \alpha^2(s) ds }$ (\cite{jiang2005note}), 
we have 
$$\mathbb{E} I(t)  \leq e^{(\beta^*-\alpha)t} \leq \tilde{M}= \max\{e^{(\beta^*-\alpha)\bar{T}}, 1\}, \qquad  t\in [0, \bar{T}].$$

Using this result,  we now consider the boundedness of $S(t)$. Taking integration and expectation of the above quation for $S$, we have
\begin{displaymath}
    \begin{aligned}
   \mathbb{E}  S(t) & \leqslant \int_{0}^{t} \{a_2\mathbb{E} S(t)-h\mathbb{E} S^2(t)\} \, d t +\gamma \int_{0}^{t} \mathbb{E} (I(t))dt  \\
   & \leqslant \int_{0}^{t} \{a_2\mathbb{E} S(t)-h\mathbb{E}^2 S(t)\} \, d t +\gamma \int_{0}^{t} \tilde{M} dt,  \qquad  t\in [0, \bar{T}]. \\
     \end{aligned}
 \end{displaymath}
Therefore,  $\mathbb{E} S(t)\leq y(t), t \in [0, \bar{T}]$, where $y$ is the solution of this ordinary differential equation:
\begin{equation}\label{mulequa14}
  \left\{\begin{aligned}
    &d y(t) = \{a_2y(t)-hy^2(t) +\gamma  \tilde{M}\}d t, \\
    &y(0)=S_0>0,    \qquad  t\in [0, \bar{T}].
  \end{aligned}\right.
\end{equation}
It is easy see that there exists $ \tilde{\tilde{M} }>0$ such that  $y(t) \leq \tilde{\tilde{M} } $ in $[0,\bar{T}].$
Therefore, we have $\mathbb{E} S(t)\leqslant \tilde{\tilde{M} }$ in $[0,\bar{T}]$.


Let us now prove  the global positivity of $x$, $S$, and $I$.
We use the same method in Case 4 to show that $\tau^\ast =\infty$. 
Consider a positive integer $k_0$ such that $x_0, I_0$ lie in $[\frac{1}{k_0}, k_0]$. Define 
$$H_k^2 = [\frac{1}{k}, k]\times [\frac{1}{k}, k],$$
 and 
\begin{displaymath}
  \tau_k=   \inf \{ 0<t<\tau^\ast; (x(t),I(t))\notin H_k\}
\end{displaymath}
The increasing sequence $\{\tau_k\}_{k=k_0}^\infty$ has a limit $\tau_\infty$ satisfying
\begin{displaymath}
  \tau_\infty = \lim_{k\to\infty} \tau_k \leqslant \tau^\ast  \qquad a.s.
\end{displaymath}
Therefore, it suffices to show that $ \tau_\infty =\infty$  a.s. 
We suppose the contrary. Then, there exist $T>0$, and $0<\varepsilon <1$, such that
\begin{displaymath}
  \mathbb{P} \{\tau_\infty<T\} >\varepsilon. 
\end{displaymath}

Consider a positive function
\begin{displaymath}
  H(x, I)=x^2+I^2-\log x-\log I =V_1+V_3,
\end{displaymath}
where $V_1$ and $V_3$ are defined by \eqref{E1}. By the It$\hat{o} $ formula, we have
  \begin{align*}
    dV_1 =&[2x^2(t)(a_1-b_1x(t))-(a_1-b_1x(t))-\frac{2 ax^2(t) (S(t)+I(t))}{b+S(t)+I(t)+x(t)}+\frac{a(S(t)+I(t))}{b+S(t)+I(t)+x(t)}    \\
    &+\sigma_1^2x^2(t)+\frac{1}{2}\sigma_1^2]dt + \sigma_1x(t)[2x(t)-\frac{1}{x(t)}]dW_1, 
     \end{align*}
and
  \begin{align*}
  dV_3  =& [2\beta(S,I)S(t)I^2(t)-2(\alpha-\gamma) I^2(t)-2\gamma I^2(t)-\beta(S,I)S(t)+(\alpha-\gamma)+\gamma+\sigma_3^2I^2(t) \\
&+\frac{1}{2}\sigma_3^2]dt +(2\sigma_3I^2(t)-\sigma_3)d W_3.
     \end{align*}

Using the same arguments  in Case 4 (see \eqref{eeq1} and \eqref{eeq3}),
it is clear that there exist $K_1, K_2 >0$ such that  
$$
dV_1 \leqslant  (K_1V_1+K_2)dt +\sigma_1x(t)(2x(t)-\frac{1}{x(t)}) dW_1,  \qquad t \in [0, \tau_k),
$$  
and
$$
dV_3 \leqslant  [(K_1V_3+K_2)]dt +\sigma_3I(t)(2I(t)-\frac{1}{I(t)})dW_3,   \qquad t \in [0, \tau_k).
$$
Thus, we have
\begin{displaymath}
  \begin{aligned}
    dH(t) &\leqslant  [K_1H(t)+2K_2]dt +\sigma_1x(t)(2x(t)-\frac{1}{x(t)} )dW_1 +\sigma_3I(t)(2I(t)+\frac{1}{I(t)} )dW_3. & \\
  \end{aligned}
\end{displaymath}
Therefore,
\begin{displaymath}
  \begin{aligned}
  \int_{0}^{t\wedge \tau_k}  \,dH \leqslant  &\int_{0}^{t\wedge \tau_k} (K_1H(s)+2K_2) \,ds+ \int_{0}^{t\wedge \tau_k} \sigma_1x(s)(2x(s)-\frac{1}{x(s)} )\,dW_1 \\
  &+ \int_{0}^{t\wedge \tau_k} (2\sigma_3I^2(s)-\sigma_3)\,dW_3,     \qquad       t\in [0,T].
  \end{aligned}
\end{displaymath}

 Taking expectation in the two sides, we obtain that 
 \begin{displaymath}
  \begin{aligned}
  \mathbb{E} \int_{0}^{t\wedge \tau_k}  \,dH \leqslant  & \mathbb{E} \int_{0}^{t\wedge \tau_k} (K_1H(s)+K_2) \,ds+\mathbb{E}  \int_{0}^{t\wedge \tau_k} \sigma_1x(s)(2x(s)-\frac{1}{x(s)} )\,dW_1  \\
  &+ \mathbb{E}  \int_{0}^{t\wedge \tau_k} (2\sigma_3I^2(s)-\sigma_3)\,dW_3.  \\
  \end{aligned}
\end{displaymath}
Thus
 \begin{displaymath}
  \begin{aligned}
  \mathbb{E} [V(x(t\wedge \tau_k), I(t\wedge \tau_k))-V(x_0, I_0)]&\leqslant K_2(t\wedge \tau_k)+\mathbb{E} \int_{0}^{t\wedge \tau_k} K_1H \,ds, & \\
  &\leqslant K_2T+\int_{0}^{t\wedge \tau_k} \mathbb{E} K_1H \,ds. & \\
  \end{aligned}
\end{displaymath}
According to the Gronwall inequality, we can draw a conclusion that
\begin{displaymath}
  \begin{aligned}
  \mathbb{E} H(x(t\wedge \tau_k), I(t\wedge \tau_k)) &\leqslant [H(x_0, I_0)+K_2T]e^{K_1t} & \\
  &\leqslant [H(x_0, I_0)+K_2T]e^{K_1T},   \qquad    t \in [0,T].
  \end{aligned}
\end{displaymath}

Using the same arguments in Case 4, we arrive at a contradiction that 
$$\infty \leq [H(x_0, I_0)+K_2T]e^{K_1T} <\infty.$$ 
Therefore $\tau_\infty =\infty$ a.s.,  and then  $\tau^\ast =\infty$ a.s.

Thanks to Cases 1$\sim $5, the proof is complete.
\end{proof}

In the next theorem, we show boundedness in expectation for $x$, $S$, $I$.

\begin{theorem}\label{thm:tvb}
  Let $(x(t), S(t), I(t))$ be the solution of \eqref{mulequa5} with initial data  $(x_0, S_0, I_0) \in \overline{\mathbb{R}_+^3}$. Then
\begin{itemize}
\item [(i)] For $1\leqslant \theta <\infty$, there exists $\alpha_{\theta}>0$ such that 
  \begin{displaymath}
     \mathbb{E} x^\theta (t)\leqslant \alpha_{\theta}, \quad t\in [0,\infty).
  \end{displaymath}
\item [(ii)]If  $\alpha-\beta^*>0$, then for any $1\leqslant \theta < 1+\frac{2(\alpha-\max\{\beta^*, \gamma\})}{\sigma^2_3}$, there exists a constant $\bar{\alpha}_{\theta}>0$ such that 
 \begin{displaymath}
   \mathbb ES^{\theta}(t)+ \mathbb E I^{\theta}(t)  \leq  \bar{\alpha}_{\theta}, \quad t\in [0,\infty).
  \end{displaymath}
\end{itemize}
\end{theorem}

\begin{proof} 
Let us first prove (i).
From 
\begin{displaymath}
 \begin{aligned}
 d x(t)=&\left[ x(t)\{a_1-b_1x(t)\}-\frac{ax(t)(S(t)+I(t))}{b+S(t)+I(t)+x(t)} \right]d t+ \sigma_1x(t)dW_1(t),
 \end{aligned}
\end{displaymath}
and the It$\hat{o}$ formula, we have for $\theta \geq 1$
\begin{displaymath}
 \begin{aligned}
 d x^{\theta}(t)=&[ \theta x^{\theta}(t)\{a_1-b_1x(t)\}-\frac{a\theta x^{\theta}(t)(S(t)+I(t))}{b+S(t)+I(t)+x(t)}+\frac{1}{2}\theta (\theta-1)\sigma^2_1x^{\theta}(t)]d t+ \sigma_1\theta x^{\theta}(t)dW_1(t), &\\
 \leqslant&   \theta x^{\theta}(t)\{a_1-b_1x(t)\}+\frac{1}{2}\theta (\theta-1)\sigma^2_1x^{\theta}(t)]d t+ \sigma_1\theta x^{\theta}(t)dW_1(t). &\\
 \end{aligned}
\end{displaymath}
Taking integration and expectation in the two sides, we obtain
\begin{displaymath}
 \begin{aligned}
\mathbb{E} x^{\theta}(t) \leqslant& \int_{0}^{t} [a_1\theta \mathbb{E} x^{\theta}(s)-b_1\theta \mathbb{E} x^{\theta+1}(s)+\frac{1}{2}\theta (\theta-1)\mathbb{E} x^{\theta}(s)] \,ds+   \mathbb{E}\int_{0}^{t} \theta \sigma_1 x^{\theta}(s) \,dW_1.
\end{aligned}
\end{displaymath}
By using the H$\ddot{o}$lder inequality, we have
\begin{displaymath}
 \begin{aligned}
\mathbb{E} x^{\theta}(t) \leqslant& \int_{0}^{t} [a_1\theta \mathbb{E} x^{\theta}(s)-b_1\theta (\mathbb{E} x^{\theta}(s))^{\frac{\theta +1}{\theta}}+\frac{1}{2}\theta (\theta-1)\mathbb{E} x^{\theta}(s)]ds. \\
\end{aligned}
\end{displaymath}
It is easy to see that  there exists a constant $C(\theta)>0$ such that 
$$a_1\theta \mathbb{E} x^{\theta}(s)-b_1\theta (\mathbb{E} x^{\theta}(s))^{\frac{\theta +1}{\theta}}+\frac{1}{2}\theta (\theta-1)\mathbb{E} x^{\theta}(s) < C(\theta)- \mathbb{E} x^{\theta}(s).$$
Thus,
\begin{displaymath}
 \begin{aligned}
  d \mathbb{E} x^{\theta}(t) \leqslant  [C(\theta)- \mathbb{E} x^{\theta}(t)] dt.
\end{aligned}
\end{displaymath}

By the comparison theorem, $\mathbb{E} x^{\theta}(t) \leqslant y(t)$ where $y$ is the solution to this 
 ordinary differential equation
\begin{equation*}
 \begin{cases}
d y =[C(\theta)- y] dt, \\
y(0) =\mathbb{E} x^{\theta}(0).
\end{cases}
\end{equation*}
It is easy to see that there exists $\alpha_{\theta}>0$ such that  $y(t)\leqslant \alpha_{\theta}$ for every  $ t\geq 0$. Thus, 
$$ \mathbb{E} x^{\theta}(t) \leq \alpha_{\theta}, \quad t\in [0,\infty). $$

Next, let us prove  (ii). From the second and third equations of \eqref{mulequa5} and 
the It$\hat{o}$ formula, we have
\begin{displaymath}
 \begin{aligned}
  d S^{\theta}(t)=&[\theta S^{\theta}(t)(a_2-\{ \frac{1}{x(t)+e}+h \}  \{S(t)+I(t)\})-\theta \beta(S,I)S^{\theta}(t)I(t)+\theta\gamma S^{\theta-1}(t) I(t)]dt\\
 &+\frac{1}{2}\theta (\theta-1)\sigma^2_2S^{\theta}(t) dt+\theta \sigma_2S^{\theta}(t)dW_2, \\
 \end{aligned}
\end{displaymath}
and
\begin{displaymath}
 \begin{aligned}
 d I^{\theta}(t)=&\Big[\theta I^{\theta-1}(t)[\beta (S,I) S(t)I(t)- \alpha I(t) ]+ \frac{1}{2}\theta (\theta-1)\sigma^2_3I^{\theta}(t)\Big]dt + \theta \sigma_3I^{\theta}(t)dW_3. 
\end{aligned}
\end{displaymath}
Taking integration and expectation of the two sides, we obtain that
$$
  \mathbb{E} S^{\theta}(t)\leqslant  \int_{0}^{t} [\theta a_2\mathbb{E}S^{\theta}(s)-h\theta \mathbb{E}S^{\theta+1}(s)+\theta \gamma \mathbb{E}( S^{\theta-1}(s) I(s))] \,ds +\int_{0}^{t} \frac{1}{2}\theta (\theta-1)\sigma^2_2 \mathbb{E}S^{\theta}(s)\,ds, 
$$
and 
$$
\mathbb{E} I^{\theta}(t)\leqslant  \int_{0}^{t} [\theta \beta^*\mathbb{E}(I^{\theta-1}(s) S(s))-\gamma\theta \mathbb{E}I^{\theta}(s)-(\alpha-\gamma) \theta \mathbb{E} I^{\theta}(s)+\frac{1}{2}\theta (\theta-1)\sigma^2_3\mathbb{E} I^{\theta}(s) ]\,ds.
$$
By using  the H$\ddot{o}$lder inequality, we then have
$$
d  \mathbb{E} S^{\theta}(t)\leqslant  \Big[ \theta a_2\mathbb{E}S^{\theta}(t)-h\theta \{\mathbb{E}S^{\theta}(t)\}^{\frac{\theta+1}{\theta}}+\theta \gamma \{ \mathbb E S^{\theta}(t)\}^{\frac{\theta-1}{\theta}} \{\mathbb{E} I^{\theta}(t)\}^{\frac{1}{\theta}}  +\frac{1}{2}\theta (\theta-1)\sigma^2_2 \mathbb{E}S^{\theta}(t) \Big] dt, 
$$
and
$$
d  \mathbb{E} I^{\theta}(t)\leqslant  \Big[ \theta \beta^* \{\mathbb{E}(I^{\theta}(t)\}^{\frac{\theta-1}{\theta}} \{\mathbb{E}S^{\theta}(t)\}^{\frac{1}{\theta}}-\gamma\theta \mathbb{E}I^{\theta}(t)-(\alpha-\gamma) \theta \mathbb{E} I^{\theta}(t)+\frac{1}{2}\theta (\theta-1)\sigma^2_3\mathbb{E} I^{\theta}(t)  \Big] dt.
$$
Since for any $x, y>0$,
$$
\theta \gamma x^{\frac{\theta-1}{\theta}} y^{\frac{1}{\theta}} +\theta \beta^*y^{\frac{\theta-1}{\theta}} x^{\frac{1}{\theta}}\leqslant \theta \max\{\beta^*,\gamma\}(x+y),
$$
 we have
\begin{displaymath}
 \begin{aligned}
d (\mathbb{E} S^{\theta}(t)+ \mathbb{E} I^{\theta}(t))  \leq &  \Big [\theta a_2\mathbb{E}S^{\theta}(t) 
	-h\theta \{\mathbb{E}S^{\theta}(t)\}^{\frac{\theta+1}{\theta}}
	 +\frac{1}{2}\theta (\theta-1)\sigma^2_2 \mathbb{E}S^{\theta}(t) -\alpha \theta \mathbb{E} I^{\theta}(t) \\
&+\frac{1}{2}\theta (\theta-1)\sigma^2_3\mathbb{E} I^{\theta}(t)+ \theta \max\{\beta^*,\gamma\}\{\mathbb{E} S^{\theta}(t)+\mathbb{E}I^{\theta}(t)\} \Big] dt\\
\leqslant& \Big [ \bar{C}(\theta)-\mathbb{E} S^{\theta}(t)
+[\frac{1}{2}(\theta-1)\sigma^2_3-\alpha+\max\{\beta^*,\gamma\} ]\theta \mathbb{E} I^{\theta}(t)\Big] dt,
\end{aligned}
\end{displaymath}
where $\bar{C}(\theta)$ is a positive constant only depending on $\theta$ and model parameters. 
Because
$$\frac{1}{2}(\theta-1)\sigma^2_3-\alpha+\max\{\beta^*,\gamma\} <0,$$
we have
$$d (\mathbb ES^{\theta}(t)+ \mathbb E I^{\theta}(t)) \leq  [ \bar{C}(\theta)- \kappa  \{ \mathbb ES^{\theta}(t)+ \mathbb E I^{\theta}(t) \} ]dt,$$
where
$$\kappa = \min \{ 1, \alpha -\max\{\beta^*,\gamma\} -\frac{1}{2}(\theta-1)\sigma^2_3 \} >0.$$
Using the comparison theorem again for this differential inequality, we imply the boundedness of $\mathbb ES^{\theta}(t)+ \mathbb E I^{\theta}(t).$ The proof is complete.
\end{proof}



\section{Decline of species}\label{sec:dec}
This section presents sufficient conditions for the decline of both prey and predator species, which are characterized by $S(t), I(t)$ and $x(t)$ tending to 0 a.s. as $t \rightarrow \infty$. Typically, such declines occur when the intensity of noise is high. In this regard, we first investigate the predator species and identify that its intrinsic growth rate $a_2$, recovery rate $\gamma$, and mortality of infected predators $\alpha^*$ must be small, while the intensities $\sigma_i \, (i=2,3)$ must be high for the decline to occur. 

\begin{theorem}\label{thm:dos}
  Let $(x(t), S(t), I(t))$ be the solution of \eqref{mulequa5} with $(x(0), S(0), I(0))=(x_0, S_0, I_0) \in \overline{\mathbb{R}_+^3} $  $\setminus (0,0,0)$.
  Assume that 
   $$\max{\{\beta^*-\alpha, a_2\}}
   <    \frac{1}{2(\frac{1}{\sigma_2^2}+\frac{1}{\sigma_3^2})}.$$
Then, when $t\rightarrow \infty$, S(t) and I(t) converge to 0 almost surely.
\end{theorem}

  \begin{proof}
The method in this proof is similar with \cite{rajasekar2020ergodic}. 
Define a function $P(S, I)$ by 
 $$P(S, I) = \epsilon S+ I,$$
where $\epsilon>0$ is a constant such that
\begin{equation} \label{E4}
\epsilon \max\{ |\beta^* -\gamma |, \gamma\} +\max{\{\beta^*-\alpha, a_2\}}
   -  \frac{1}{2(\frac{1}{\sigma_2^2}+\frac{1}{\sigma_3^2})}<0.
   \end{equation}
We have
    \begin{displaymath}
      \begin{aligned}
      d\ln P(S(t), I(t))&= \mathcal{L} (\ln P(S(t), I(t)))dt +\frac{1}{P(S(t), I(t))}\big[\epsilon\sigma_2S(t)dW_2+\sigma_3I(t)dW_3)\big], 
      \end{aligned}
    \end{displaymath} 
where
    \begin{displaymath}
      \begin{aligned}
        \mathcal{L} &(\ln P(S(t), I(t))) \\
        =& \frac{\epsilon}{P(S(t), I(t))}\big[S(t)\big(a_2-(\frac{1}{x(t)+e} +h)\{S(t)+I(t)\}\big]-\beta(S,I)S(t)I(t)+\gamma I(t)\big]  \\
        & +  \frac{1}{P(S(t), I(t))}\big[\beta(S,I)S(t)I(t)-\alpha I(t)\big] \\
        & -\frac{\epsilon^2}{2P^2(S(t), I(t))}\sigma_2^2S^2(t) -\frac{1}{2P^2(S(t), I(t))}\sigma_3^2I^2(t) \\
        =&  Q_1-Q_2,
      \end{aligned}
    \end{displaymath}
in which
    \begin{displaymath}
      \begin{aligned}
        Q_1=&\frac{\epsilon}{P(S(t), I(t))}\big[S(t)\big(a_2-(\frac{1}{x(t)+e} +h)\{S(t)+I(t)\}\big)-\beta(S,I)S(t)I(t)+\gamma I(t)\big] & \\
        &+  \frac{1}{P(S(t), I(t))}\big[\beta(S,I)S(t)I(t)-\alpha I(t)\big],
      \end{aligned}
    \end{displaymath} 
    and 
    $$Q_2=\frac{\epsilon^2}{2P^2(S(t), I(t))}\sigma_2^2S^2(t) +\frac{1}{2P^2(S(t), I(t))}\sigma_3^2I^2(t).$$
    
We have
    \begin{displaymath}
      \begin{aligned}
        Q_1 \leq & \frac{\epsilon}{P(S(t), I(t))}\big[a_2S(t)-\beta(S,I)S(t)I(t)+\gamma I(t)\big]
        + \frac{1}{P(S(t), I(t))}\big[\beta(S,I)S(t)I(t)- \alpha I(t)\big]  \\
        \leqslant &  \frac{\epsilon}{P(S(t), I(t))}\big[a_2S(t)-\beta(S,I)S(t)I(t)+\gamma I(t)\big]
        + \frac{1}{P(S(t), I(t))}\big[\beta^*I(t)- \alpha I(t)\big]  \\
        =& \frac{\epsilon}{P(S(t), I(t))}I(t)\big[\gamma-\beta(S,I)S(t)\big]+ \frac{1}{P(S(t), I(t))}\big[\epsilon a_2S(t)+(\beta^*- \alpha )I(t)\big]  \\
        \leqslant &  \frac{\epsilon }{P(S(t), I(t))}I(t)| \gamma-\beta(S,I)S(t)\vert + \frac{1}{P(S(t), I(t))}\big[\epsilon a_2S(t)+(\beta^*- \alpha ) I(t)\big]. 
      \end{aligned}
    \end{displaymath} 
Since $\beta(S,I)S(t) \in [0,\beta^*]$, 
 $$|\gamma-\beta(S,I)S(t)\vert \leqslant \max\{| \beta^*-\gamma\vert, \gamma\}.$$
In addition,  
 \begin{displaymath}
      \begin{aligned}
       \frac{I(t)}{P(S(t),I(t))}=\frac{I(t)}{\epsilon S(t)+ I(t)}\leqslant 1
      \end{aligned}
    \end{displaymath} 
Thus,
    \begin{displaymath}
      \begin{aligned}
        Q_1 &\leqslant \epsilon \max\{| \beta^*-\gamma\vert, \gamma\}+  \frac{1}{P}\big[\epsilon a_2S(t)+(\beta^*-\alpha) I(t)\big] & \\
        &\leqslant \epsilon \max\{| \beta^*-\gamma\vert, \gamma\} + \frac{1}{P} \max\{a_2,\beta^*-\alpha\}\big[\epsilon S(t)+ I(t)\big]&\\
        &\leqslant  \epsilon max\{| \beta^*-\gamma\vert, \gamma\}+ \max\{\beta^*-\alpha, a_2\}.& 
      \end{aligned}
    \end{displaymath} 
    
For $Q_2$, from 
    \begin{displaymath}
      \begin{aligned}
        P^2(S(t), I(t)) = (\epsilon S(t)+I(t))^2\leqslant (\epsilon^2\sigma_2^2S^2(t)+ \sigma_3^2I^2(t))(\frac{1}{\sigma_2^2}+\frac{1}{\sigma_3^2} ),
      \end{aligned}
    \end{displaymath}   
we have    
$$
        Q_2 = \frac{\epsilon^2}{2P^2}\sigma_2^2S^2(t) +\frac{1}{2P^2}\sigma_3^2I^2(t) \geqslant \frac{1}{2(\frac{1}{\sigma_2^2}+\frac{1}{\sigma_3^2} )}. 
$$
    
    Therefore, we obtain that 
    \begin{displaymath}
      \begin{aligned}
        d \ln P(S(t), I(t))\leqslant  & \Big[ \epsilon \max\{| \beta^*-\gamma\vert, \gamma\}+ \max\{\beta^*-\alpha, a_2\}- \frac{1}{2(\frac{1}{\sigma_2^2}+\frac{1}{\sigma_3^2} )} \Big] dt \\
        &+ \frac{1}{P(S(t), I(t))} \big[\epsilon\sigma_2S(t)dW_2+ \sigma_3I(t)dW_3\big].
      \end{aligned}
    \end{displaymath} 
    Taking the integral on $[0,t]$, we have
    \begin{displaymath}
      \begin{aligned}
        \ln P(S(t), I(t))\leqslant  &\ln P(S(0),I(0)) +\big[ \epsilon \max\{|\beta^*-\gamma\vert, \gamma\}+ \max\{\beta^*-\alpha, a_2\}-\frac{1}{2(\frac{1}{\sigma_2^2}+\frac{1}{\sigma_3^2} )}\big]t &\\
        &+\int_{0}^{t}  \frac{1}{P(S(s), I(s))}\big[\epsilon\sigma_2S(s)\,dW_2+  \sigma_3I(s)\,dW_3\big]. & \\
        \end{aligned}
    \end{displaymath}
    Then,
    \begin{displaymath}
      \begin{aligned}
        \frac{\ln  P(S(t), I(t))}{t}\leqslant & \frac{ \ln P(S(0),I(0))}{t} +\big[ \epsilon \max\{|\beta^*-\gamma\vert, \gamma\}+ \max\{\beta^*-\alpha, a_2\}-\frac{1}{2(\frac{1}{\sigma_2^2}+\frac{1}{\sigma_3^2} )}\big]\\
        &+ \frac{1}{t} \int_{0}^{t}  \frac{1}{P(S(s), I(s))}\big[\epsilon\sigma_2S(s)\,dW_2+  \sigma_3I(s)\,dW_3\big].  & 
      \end{aligned}
    \end{displaymath} 
    
    Let us consider two real-valued continuous martingales vanishing at $t=0$:
    $$B^{(2)}_t= \int_{0}^{t} \frac{\epsilon\sigma_2S(s)}{P(S(s),I(s))} \, dW_2(s), \quad B^{(3)}_t= \int_{0}^{t} \frac{ \sigma_3I(s)}{P(S(s),I(s))} \, dW_3(s).$$ 
    They have quadratic form given by 
$$
       \langle B^{(2)}_t\rangle = \int_{0}^{t} \frac{\xi^2_1\sigma^2_2S^2(s)}{P^2(S(s),I(s))} \, ds  \leq \sigma_2^2 t,  \qquad
        \langle B^{(3)}_t\rangle = \int_{0}^{t} \frac{\xi^2_2\sigma^2_3I^2(s)}{P^2(S(s),I(s))} \, ds \leq \sigma_3^2 t.
$$
    Using the strong law of large numbers for martingale\cite{mao2007stochastic}, we have
    \begin{displaymath}
      \begin{aligned}
        \lim_{t\to\infty}\frac{B^{(i)}_t}{t}=0  \qquad a.s. \qquad (i=2,3).
      \end{aligned}
    \end{displaymath}
    Thus, we obtain that
    \begin{displaymath}
      \begin{aligned}
       \limsup_{t\to\infty} \frac{ \ln P(S(t), I(t))}{t}  \leqslant  & \lim_{t\to\infty}\frac { \ln P(S(0),I(0))}{t} +\big[ \epsilon  \max\{|\beta^*-\gamma\vert, \gamma\} \\
       &+  \max\{\beta^*-\alpha, a_2\}-\frac{1}{2(\frac{1}{\sigma_2^2}+\frac{1}{\sigma_3^2} )}\big]  \\
       \leqslant & \epsilon  \max \{|\beta^*-\gamma\vert, \gamma\}+  \max \{\beta^*-\alpha, a_2\}-\frac{1}{2(\frac{1}{\sigma_2^2}+\frac{1}{\sigma_3^2} )} \qquad  a.s. 
      \end{aligned}
    \end{displaymath} 
It then follows from \eqref{E4} that 
    \begin{displaymath}
      \begin{aligned}
        \limsup _{t\to\infty} \frac{ \ln P(S(t), I(t))}{t}&<0 \qquad  a.s.
      \end{aligned}
    \end{displaymath} 
    As a consequence, we have 
    \begin{displaymath}
      \begin{aligned}
        \limsup _{t\to\infty} \frac{ \ln S(t)}{t}&<0, \qquad
        \limsup _{t\to\infty} \frac{ \ln I(t)}{t}&<0 \qquad a.s.&
      \end{aligned}
    \end{displaymath}
These imply that 
    \begin{displaymath}
      \begin{aligned}
        \lim_{t\to\infty} S(t)=0,\qquad  \lim_{t\to\infty} I(t)=0 \qquad a.s.,
      \end{aligned}
    \end{displaymath} 
i.e.  the population of predator declines.
  \end{proof}

 In the next theorem, we consider the decline of prey species.
 
\begin{theorem}\label{thm:dop}
  Let $(x(t), S(t), I(t))$ be the solution of \eqref{mulequa5} with $(x(0), S(0), I(0))=(x_0, S_0, I_0) \in \overline{\mathbb{R}_+^3}\setminus (0, 0, 0)$. 
  Assume that $a_1-\frac{1}{2}\sigma_1^2 <0$. Then,
  when $t\rightarrow \infty$, x(t) converges to 0 almost surely.
\end{theorem}
\begin{proof}
  From
  \begin{displaymath}\label{mulequa6}
    \begin{aligned}
    d x(t)&=\big\{x(t)\big[a_1-b_1x(t)\big]-\frac{ax(t)(S(t)+I(t))}{b+S(t)+I(t)+x(t)}\big\}d t+ \sigma_1x(t)dW_1(t) 
  \end{aligned}
\end{displaymath}
we have
\begin{displaymath}\label{mulequa7}
  \begin{aligned}
    d\ln x(t)&=\big\{(a_1-b_1x(t))-\frac{a(S(t)+I(t))}{b+S(t)+I(t)+x(t)}- \frac{1}{2}\sigma_1^2\big\}d t+ \sigma_1dW_1(t) \\
    & \leqslant (a_1 -\frac{1}{2}\sigma_1^2)d t + \sigma_1dW_1(t). 
    \end{aligned}
  \end{displaymath}
  Therefore,
  $$ \frac{\ln x(t)-\ln x_0}{t} \leqslant (a_1-\frac{1}{2}\sigma_1^2)+\frac{\sigma_1W_1}{t}   \qquad a.s.   $$
  Because $\lim_{t \to \infty} \frac{\sigma_1W_1}{t} =0$, taking $t \rightarrow \infty$ in the later inequality yields
  
  \begin{displaymath}\label{mulequa8}
    \begin{aligned}
      \limsup_{t\to\infty}\frac{\ln x(t)}{t} \leqslant a_1-\frac{1}{2}\sigma_1^2<0.
    \end{aligned}
  \end{displaymath}
  As a consequence, $\lim_{t\to\infty} x(t)=0$ a.s. The proof is complete.
\end{proof}

\section{Sustainability of species}\label{sec:scs}
In this section, we present some sufficient conditions for the sustainability of the system. These conditions demonstrate that the system can be sustainable if the intensity of noise and the mortality of infected predators are high, and the intrinsic growth rate of prey and predator species is large.

Firstly, we define the sustainability of \eqref{mulequa5} as follows.
\begin{definition}
  The system \eqref{mulequa5} is sustainable if, for every initial value $(x(0), S(0), I(0)) \in \overline{\mathbb{R}_+^3}\setminus (0, 0, 0)$, the 
  solution $(x(t),S(t),I(t))$ satisfies
  \begin{displaymath}
    \begin{aligned}
      \liminf_{t\to\infty} \frac{1}{t} \int_{0}^{t} S(s) \,ds >0,\qquad  \liminf_{t\to\infty} \frac{1}{t} \int_{0}^{t} x(s) \,ds >0, \qquad  \lim_{t\to\infty} I(t)=0  \qquad a.s.
    \end{aligned}
  \end{displaymath}
\end{definition}

To show the sustainability of the system, we use the following lemma.
\begin{lemma}\cite{liu2011survival, mandal2012stochastic}\label{lem: psm}
  Let $\{B_i(t)\}_{i=1,2,\dots,N}$ be  sequence of standard Brownian motions and $x \in \mathcal{C} [\mathbb{\overline{R}}_+, \mathbb{R}_+^0]$  a random process on
   on the probability space $(\Omega, \mathcal{F}, \mathbb{P} )$. 
\begin{itemize}
\item[(i)] If there are positive constants $\mu $, $\lambda,$ and a random time $T>0$ such that
  \begin{displaymath}
    \ln x(t) \leqslant \lambda t-\mu \int_{0}^{t} x(s) \,ds +\sum_{i = 1}^{n} \beta_iB_i(t)  
  \end{displaymath}
  for $t\geqslant T$ a.s., where $\beta_i$ are constants,  then 
  $$\limsup_{t\to\infty} \frac{1}{t} \int_{0}^{t} x(s) \,ds\leqslant \frac{\lambda}{\mu} \quad a.s.$$
\item[(ii)] If there are positive constants $\mu $, $\lambda,$ and a random time $T>0$ such that
  \begin{displaymath}
    \ln x(t) \geqslant  \lambda t-\mu \int_{0}^{t} x(s) \,ds +\sum_{i = 1}^{n} \beta_iB_i(t) 
  \end{displaymath}
  for $t\geqslant T$ a.s., where $\beta_i$ are constants,  then 
  $$\liminf_{t\to\infty} \frac{1}{t} \int_{0}^{t} x(s)  \,ds \geqslant  \frac{\lambda}{\mu} \quad a.s.$$
\end{itemize}
\end{lemma}

We are now ready to state our main results in this section.
\begin{theorem}\label{thm:scs}
Let $(x(t),S(t),I(t))$ be the solution of \eqref{mulequa5} with initial value $(x_0,S_0,I_0) \in \overline{\mathbb{R}_+^3}\setminus (0,0,0).$ Assume that
\begin{equation}\label{mulequa17}
  \left\{\begin{aligned}
    &a_1-a-\frac{1}{2}\sigma_1^2>0, \\
    &a_2-\beta^*-\frac{1}{2}\sigma_2^2>0,     \\ 
    &\beta^*-\alpha-\frac{1}{2}\sigma_3^2 <0.
  \end{aligned}\right.
\end{equation}
Then, 
$$\lim_{t\to\infty} I(t) =0 \quad a.s.$$ 
 Furthermore,
there exists $\varepsilon >0$ which is independent of initial value such that
\begin{displaymath}
  \begin{aligned}
    &\liminf_{t\to\infty} \frac{1}{t} \int_{0}^{t} x(s) \,ds >\varepsilon, \qquad  \liminf_{t\to\infty} \frac{1}{t} \int_{0}^{t} S(s) \,ds >\varepsilon.
  \end{aligned}
\end{displaymath}
\end{theorem}
\begin{proof}
  We firstly show the extinction of infected population, i.e. $\lim_{t\to\infty} I(t) =0$. 
From
  \begin{displaymath}
    \begin{aligned}
      d I(t)&=\big[\beta(S,I)S(t)I(t)- \alpha I(t)\big]d t+ \sigma_3I(t)dW_3(t),
    \end{aligned}
  \end{displaymath}
  we have 
  \begin{displaymath}
    \begin{aligned}
      d \ln I(t)&=(\beta(S,I)S(t)- \alpha  -\frac{1}{2} \sigma_3^2)d t +\sigma_3dW_3(t)  \\
      &\leqslant (\beta^*-\alpha -\frac{1}{2} \sigma_3^2)d t +\sigma_3dW_3(t). 
    \end{aligned}
  \end{displaymath}
  Then,
  \begin{displaymath}
    \begin{aligned}
      \frac{\ln I(t)-\ln I_0}{t} &\leqslant (\beta^* -\alpha -\frac{1}{2} \sigma_3^2)+ \frac{\sigma_3W_3(t)}{t}  \qquad a.s.& 
    \end{aligned}
  \end{displaymath}
Thus, 
$$\lim_{t \to \infty}  \frac{\ln I(t)}{t} \leqslant \beta^* -\alpha -\frac{1}{2} \sigma_3^2 <0 \qquad a.s. $$
As a result, we obtain that 
$$\lim_{t \to \infty} I(t) =0  \quad a.s. $$

 Secondly, we consider $S(t)$. Since 
\begin{displaymath}
  \begin{aligned}
    d S(t)&=\big(S(t)\big[a_2-(\frac{1}{x(t)+e} +h)\{S(t)+I(t)\}\big]-\beta(S,I)S(t)I(t)+\gamma I(t) \big) d t +\sigma_2S(t)dW_2(t), &
  \end{aligned}
\end{displaymath}
we have
\begin{displaymath}
  \begin{aligned}
    d \ln S(t)&=\big[a_2-(\frac{1}{x(t)+e} +h)(S(t)+I(t))-\beta(S,I)I(t)+\gamma \frac{ I(t) }{S(t)} - \frac{1}{2} \sigma_2^2\big]d t +\sigma_2dW_2(t) & \\
    &\geqslant \big[a_2-(\frac{1}{x(t)+e} +h)(S(t)+I(t))-\beta(S,I)I(t) - \frac{1}{2} \sigma_2^2\big]d t +\sigma_2dW_2(t) & \\
    &\geqslant \big[a_2-(\frac{1}{e} +h)(S(t)+I(t))-\beta^*- \frac{1}{2} \sigma_2^2\big] d t +\sigma_2dW_2(t).& 
  \end{aligned}
\end{displaymath}
Therefore,
\begin{displaymath}
  \begin{aligned}
    \frac{\ln S(t)-\ln S_0 }{t} &\geqslant (a_2-\beta^*-\frac{1}{2}\sigma_2^2) -\frac{1}{t}\int_{0}^{t} (\frac{1}{e}+h)S(s)  \,ds - \frac{1}{t}\int_{0}^{t} (\frac{1}{e}+h)I(s)  \,ds + \frac{\sigma_2W_2(t)}{t}  \qquad a.s.& 
  \end{aligned}
\end{displaymath}
Since $\lim_{t\to \infty} I(t) =0 $ a.s. and $a_2-\beta^*-\frac{1}{2}\sigma_2^2>0$, there exists a random $T>0$ such that for all $ t \geqslant T$
\begin{displaymath}
  \begin{aligned}
    \frac{1}{t}\int_{0}^{t} (\frac{1}{e}+h)I(s)  \,ds < \frac{1}{2} (a_2-\beta^*-\frac{1}{2}\sigma_2^2) \quad a.s.
  \end{aligned}
\end{displaymath}
Thus, for all $t>T$,
\begin{equation}  \label{E5}
  \begin{aligned}
    \frac{\ln S(t)-\ln S_0 }{t} 
		&\geqslant (a_2-\beta^*-\frac{1}{2}\sigma_2^2) -\frac{1}{t}\int_{0}^{t} (\frac{1}{e}+h)S(s)  \,ds - \frac{1}{2} (a_2-\beta^*-\frac{1}{2}\sigma_2^2) + \frac{\sigma_2W_2(t)}{t}  \\
   		 &= \frac{1}{2} (a_2-\beta^*-\frac{1}{2}\sigma_2^2)-\frac{1}{t}\int_{0}^{t} (\frac{1}{e}+h)S(s)  \,ds+ \frac{\sigma_2W_2(t)}{t},
  \end{aligned}
\end{equation}
or equivalently
$$
   \ln z(t) \geqslant   \frac{1}{2}(a_2-\beta^*-\frac{1}{2}\sigma_2^2)t -  (\frac{1}{e}+h) S(0) \int_{0}^{t}z(s)  \,ds +\sigma_2W_2(t),
$$
where $z(t)=\frac{S(t)}{S(0)}$.

Using Lemma \ref{lem: psm} for the latter inequality, we obtain that
\begin{displaymath}
  \begin{aligned}
    \liminf_{t\to\infty} \frac{1}{t} \int_{0}^{t} z(s) \,ds &\geqslant \frac{a_2-\beta^*-\frac{1}{2}\sigma_2^2}{2(\frac{1}{e}+h) S(0)} >0  \qquad a.s.
  \end{aligned}
\end{displaymath}
Therefore,
$$ \liminf_{t\to\infty} \frac{1}{t} \int_{0}^{t} S(s) \,ds  \geqslant \frac{a_2-\beta^*-\frac{1}{2}\sigma_2^2}{2(\frac{1}{e}+h)} >0  \qquad a.s.$$

Finally, we consider $x(t)$. From
\begin{displaymath}
  \begin{aligned}
  d x(t)&=\big\{x(t)(a_1-b_1x(t))-\frac{ax(t)(S(t)+I(t))}{b+S(t)+I(t)+x(t)}\big\}d t+ \sigma_1x(t)dW_1(t), &
\end{aligned}
\end{displaymath}
we have
\begin{displaymath}
  \begin{aligned}
  d\ln x(t)&=\big\{a_1-b_1x(t)-\frac{a(S(t)+I(t))}{b+S(t)+I(t)+x(t)}- \frac{1}{2}\sigma_1^2\big\}d t+ \sigma_1dW_1(t) &\\
  &\geqslant  \big\{a_1-b_1x(t)-a- \frac{1}{2}\sigma_1^2\big\}d t+ \sigma_1dW_1(t). &
\end{aligned}
\end{displaymath}
By using the same argument as above, we have
$$ \liminf_{t\to\infty} \frac{1}{t} \int_{0}^{t} x(s) \,ds  \geqslant \frac{a_1-a-\frac{1}{2}\sigma_1^2}{b_1} >0  \qquad a.s.$$
 The proof is complete.
\end{proof}

\section{Existence of Borel invariant measure}
In this section, we show the existence of a Borel invariant measure for the It$\hat{o}$ process $(x(t), S(t), I(t))$ on the domain of $\mathbb{R}_+^3$, assuming that the sustainability condition is satisfied.   The transition probability of $(x(t),S(t),I(t))$ is denoted by $P(.,.,.,.,.)$, i.e.,
\begin{displaymath}
  \begin{aligned}
    P(t,\xi,\eta,\varsigma, \mathcal{K} )=\mathbb{P} \{(x(t),S(t),I(t))\in \mathcal{K};(x(0),S(0),I(0))=(\xi,\eta,\varsigma)\}
  \end{aligned}
\end{displaymath}
for $0\leqslant t<\infty, (\xi,\eta,\varsigma)\in \overline{\mathbb{R}_+^3}$ and $\mathcal{K} \in \mathcal{B} (\overline{\mathbb{R}_+^3})$. 

Following \cite{lin1970conservative, foguel1973ergodic}, we have 
\begin{itemize}
\item[(i)]  $P(t, x, S, I) $ induces a strongly continuous semi-group $\{P_t\}_{0\leqslant t<\infty}$ of operators
on the space $\mathcal{C}_B(\overline{\mathbb{R}_+^3})$ of bounded continuous functions:
\begin{displaymath}
  \begin{aligned}
    P_t f(x,S,I)=\int_{\mathbb{R}_+^3}f(\xi,\eta,\varsigma  )P(x(t),S(t),I(t), d\xi d\eta d\varsigma ),  \qquad f\in \mathcal{C}_B(\overline{\mathbb{R}_+^3}).
  \end{aligned}
\end{displaymath}
\item[(ii)] $P(t,\xi, \eta,\varsigma$,.) induces a positive contraction  $[.P_t]$ on the space $M(\overline{\mathbb{R}_+^3}, \mathcal{B} (\overline{\mathbb{R}_+^3}))$ of finite 
signed measures:
\begin{displaymath}
  \begin{aligned}
    [\mu P_t](\mathcal{K} )=\int_{\mathbb{R}_+^3} P(t, \xi,\eta,\varsigma, \mathcal{K})\mu(d\xi_1d\eta_1d\varsigma_1),  \qquad   \mu \in M(\overline{\mathbb{R}_+^3}, \mathcal{B} (\overline{\mathbb{R}_+^3})), \mathcal{K} \in \mathcal{B} (\overline{\mathbb{R}_+^3}).
  \end{aligned}
\end{displaymath}
\end{itemize}

\begin{definition}
  A Borel measure $\nu$ on $ \overline{\mathbb{R}_+^3}$ (i.e. a positive measure that is finite on any compact set of $ \overline{\mathbb{R}_+^3}$)
  is said to be invariant with respect to $\{P_t\}_{0\leqslant t<\infty}$ if for $0\leqslant t<\infty$ and $\mathcal{K} \in \mathcal{B} (\overline{\mathbb{R}_+^3})$,
  \begin{displaymath}
    \begin{aligned}
     [\nu P_t](\mathcal{K} )=\nu (\mathcal{K} ).
    \end{aligned}
  \end{displaymath}
  \end{definition}
  
  The following result is well-known.
\begin{lemma}[\cite{lin1970conservative}]\label{lem:ebi}
   Let X be a locally compact perfectly normal topological space. Let $\{Q_t\}_{0\leqslant t<\infty}$ be a strongly continuous semi-group on $\mathcal{C}_B(X)$ generated by a transition 
   probability on $(X, \mathcal{B} (X))$. If there exists a non-negative function g in the space $\mathcal{C}_0(X)$ of continuous functions with compact support such
   that
   \begin{displaymath}
    \begin{aligned}
    \int_{0}^{\infty} Q_tg(t) \,dt = \infty,   \qquad x\in X,
    \end{aligned}
  \end{displaymath}
  then there exists a Borel invariant measure for $\{Q_t\}_{0\leqslant t< \infty}$.
\end{lemma}

Let us now state our main theorem in this section. 
\begin{theorem}\label{thm:ebi}
  Assume that sustainable condition \eqref{mulequa17} holds true. If $\alpha>\beta^*,$ then $\{P_t\}_{0\leqslant t<\infty}$ has a Borel invariant measure on $ \overline{\mathbb{R}_+^3}$.
\end{theorem}

  \begin{proof}
It follows from Theorem \ref{thm:tvb} that  there exists $\kappa>0$ such that 
    \begin{equation*}  
    \sup_{t\geq 0}\mathbb E x(t) \leq \kappa,\quad
      \sup_{t\geq 0} \mathbb{E} S (t)\leqslant \kappa, \quad    \sup_{t\geq 0} \mathbb{E} I (t)\leqslant \kappa.
    \end{equation*}
Let  $\varepsilon \in (0, 1) $ and  $M= \frac{4 \kappa}{\varepsilon}.$
The Markov inequality provides that for every $t\geq 0,$
    \begin{equation}  \label{E7}
\begin{cases}
        \mathbb{P}\{x(t) > M \}\leqslant \frac{\mathbb{E} x(t)}{M} \leq \frac{\kappa}{M}  =\frac{\varepsilon}{4},\\
      \mathbb{P}\{S(t) > M \}\leqslant \frac{\mathbb{E} S(t)}{M_2} \leq \frac{\kappa}{M}  =\frac{\varepsilon}{4},\\       
       \mathbb{P}\{I(t) > M \}\leqslant \frac{\mathbb{E} I(t)}{M_3} \leq \frac{\kappa}{M}  =\frac{\varepsilon}{4},
\end{cases}
    \end{equation}
and therefore
\begin{equation} \label{E12}    
\begin{cases}     
   \inf_{t \geq 0} \mathbb{P}\{0 \leqslant x(t) \leqslant M \} \geqslant 1-\frac{\varepsilon}{4}, \\
 \inf_{t \geq 0} \mathbb{P}\{0 \leqslant S(t) \leq  M \}  \geq 1-\frac{\varepsilon}{4},\\
 \inf_{t \geq 0} \mathbb{P}\{0 \leqslant I(t)\leqslant M \}  \geq 1-\frac{\varepsilon}{4}.
\end{cases}
 \end{equation}

Define a set $K$ by
    $$K=[0, M] \times [0,M] \times [0,M].$$
Using \eqref{E7} and \eqref{E12}, we have for every $t\geq 0,$ 
 \begin{align}
   \mathbb{P}\{(x(t),S(t),I(t))\in K \} = & \mathbb{P}\{0 \leqslant I(t) \leqslant M_3 \} - \mathbb{P}\{x(t) > M, 0 \leqslant S(t) \leqslant M, 0 \leqslant I(t) \leqslant M \}    \notag\\
   &- \mathbb{P}\{0 \leqslant x(t) \leqslant M, S(t) > M, 0 \leqslant I(t) \leqslant M \}  \notag\\
& - \mathbb{P}\{x(t) > M, S(t)> M, 0 \leqslant I(t) \leqslant M \}   \notag\\  
\geq & (1-\frac{\varepsilon}{4}) - \frac{\varepsilon}{4}-\frac{\varepsilon}{4}-\frac{\varepsilon}{4} =1-\varepsilon.  \label{E15}
   \end{align}

Consider a non-negative function $g\in \mathcal{C}_0(\overline{\mathbb{R}_+^3})$ such that
    \begin{displaymath}
     g(x,S,I)=
     \begin{cases}
      1, & (x,S,I)\in K,\\
      0, & (x,S,I)\notin \overline{\mathbb{R}_+^3} \setminus K_1,
     \end{cases}
    \end{displaymath}
where $K_1 \supset K$ is a bounded set of $\overline{\mathbb{R}_+^3}.$
By using \eqref{E15},  we have
   \begin{displaymath}
    \begin{aligned}
      \int_{0}^{t} P_sg(x,S,I) \,ds &= \int_{0}^{t} \int_{\overline{\mathbb{R}_+^3}} g(\xi, \eta, \varsigma )P(s, x, S, I,d\xi d\eta d\varsigma ) \,ds & \\
      & \geqslant \int_{0}^{t} \int_K g(\xi, \eta, \varsigma )P(s, x, S, I,d\xi d\eta d\varsigma ) \,ds & \\
      &= \int_{0}^{t} \mathbb{P}\{(x(s) ,S(s) ,I(s) )\in K\} \,ds \rightarrow \infty \qquad when \quad t \rightarrow \infty.
    \end{aligned}
  \end{displaymath}
Thanks to Lemma \ref{lem:ebi}, we conclude that there exists a Borel invariant measure $\nu$ for  $\{P_t\}_{0\leqslant t <\infty}$ such that $\nu(K)>0$. The proof is complete.
  \end{proof}

\section{Numerical examples}
This section exhibits some numerical examples for sustainability of the model \eqref{mulequa5} and possibility of decline. For the computations, we used a scheme of order $1.5$ (see, e.g., \cite{Kloeden2003}).

\subsection{Sustainability}
We set the parameters as follows: $a=0.6$, $b=2$, $a_1=2.7$, $a_2=1.2$, $b_1=0.1$, $c=1$, $h=0.01$, $\beta^*=1$, $k=0.1$, $\gamma=0.3$, $\alpha=1.1$, $\sigma_1=0.5$, $\sigma_2=0.6$, $\sigma_3=0.4$, and take the initial values $x(0)=6$, $S(0)=10$, and $I(0)=2$. Furthermore, we choose $N=500$ and $T=20$. Note that these parameter values satisfy the conditions of Theorem \ref{thm:scs}.

Figure \ref{figure1} displays sample trajectories of $x$, $S$, and $I$ in the phase space and in time. We calculate $N$ points in the time interval $(0,T]$ with a step size of $\frac{T}{N}$.

\begin{figure}[H]
  \begin{center}  
    \includegraphics[scale=0.7]{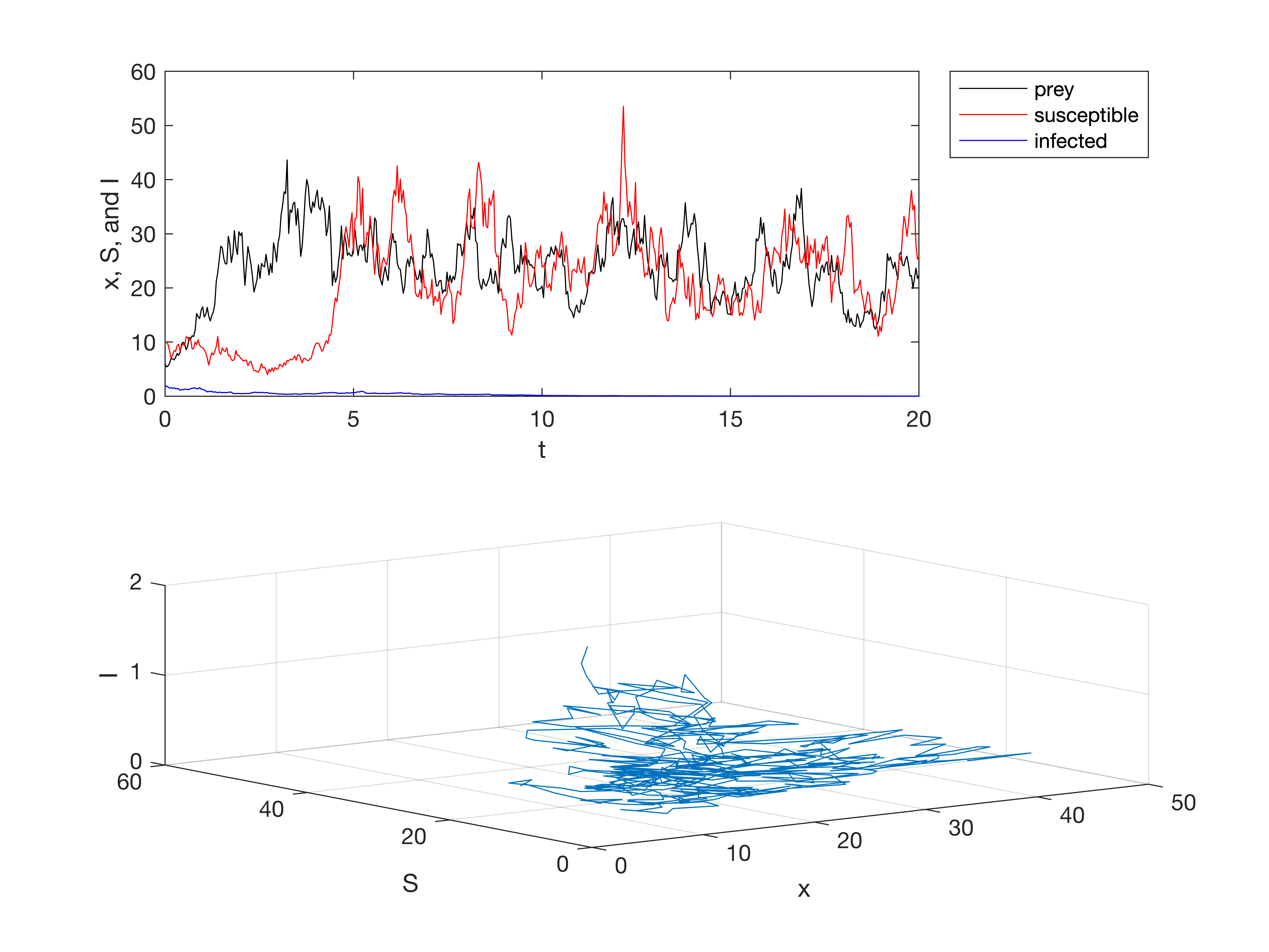}
    \caption{Sample trajectories in time (above) and in the phase space (below) with $a=0.6$, $b=2$, $a_1=2.7$, $a_2=1.2$, $b_1=0.1$, $c=1$, $h=0.01$, $\beta^*=1$, $k=0.1$, $\sigma_1=0.5$, $\sigma_2=0.6$, $\sigma_3=0.4$, $\gamma=0.3$, $\alpha=1.1$, and initial values $x(0)=6$, $S(0)=10$, and $I(0)=2$.}  \label{figure1}
  \end{center}
\end{figure}

To explore the sustainability of the system, we set new values for $N$ and $T$, as $N=5000$ and $T=50$. We keep the other parameters the same as before. We define the time-averaged quantities 
$$x^*(t) =\frac{1}{t} \int_0^t x(s)ds, \quad S^*(t) =\frac{1}{t} \int_0^t S(s)ds, \quad t \in [0,\infty),$$
with a convention that $x^*(0)=x(0)$ and $ S^*(0)=S(0).$ Figure \ref{figure2} presents a sample trajectory of the three processes $x^*$, $S^*$, and $I$. It is observed that $x^*$ and $S^*$ are bounded below by some positive constant, while $I$ converges to zero.
\begin{figure}[H]
  \begin{center}  
    \includegraphics[scale=0.7]{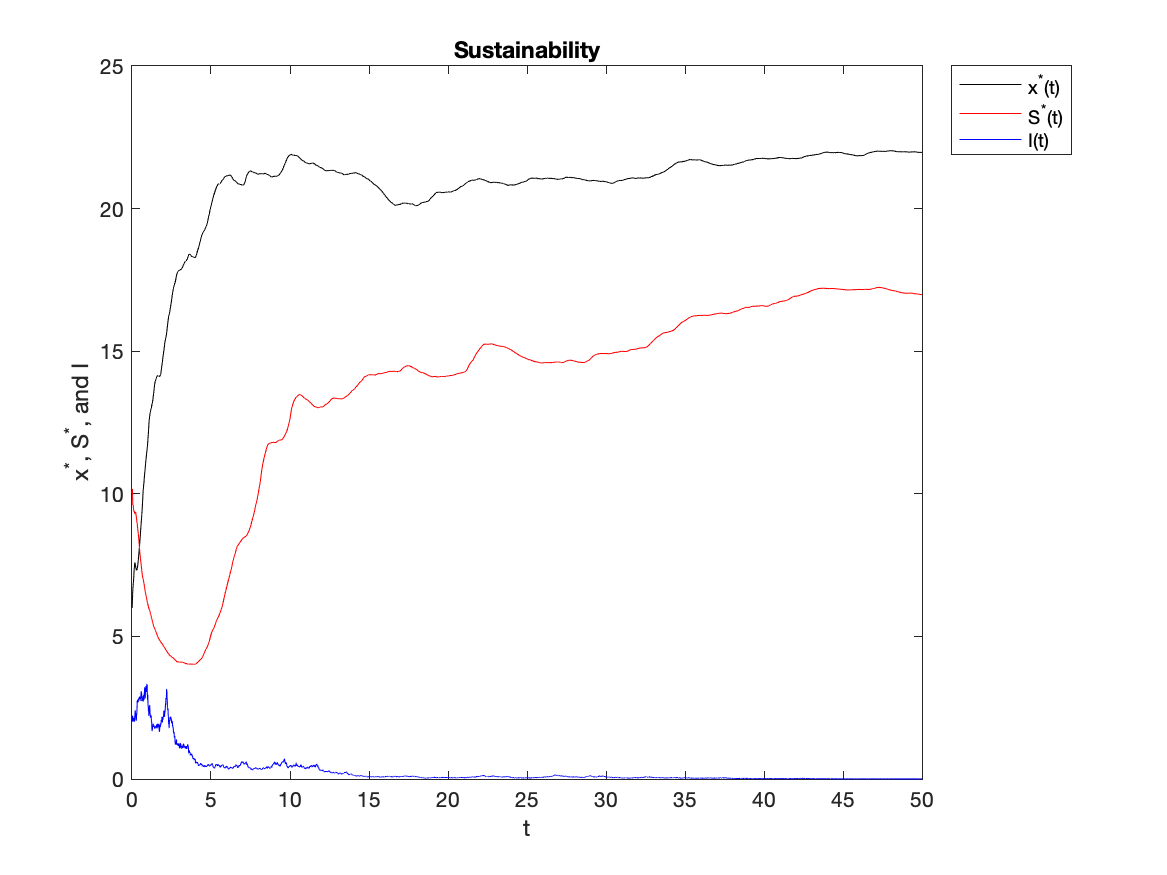}
    \caption{Sample trajectories of $x^*$, $S^*$, and $I$ with $a=0.6$, $b=2$, $a_1=2.7$, $a_2=1.2$, $b_1=0.1$, $c=1$, $h=0.01$, $\beta^*=1$, $k=0.1$, $\sigma_1=0.5$, $\sigma_2=0.6$, $\sigma_3=0.4$, $\gamma=0.3$, $\alpha=1.1$, and initial values $x(0)=6$, $S(0)=10$, and $I(0)=2$.}  \label{figure2}
  \end{center}
\end{figure}

To visualize the support of the invariant measure, we generate $10000$ samples of the process $(x,S,I)$ and plot their values at time $T$ and $T-\frac{T}{N}$ in Figure \ref{figure3}. The parameter values are the same as in Figure \ref{figure2}.
\begin{figure}[H]
  \begin{center}  
    \includegraphics[scale=0.7]{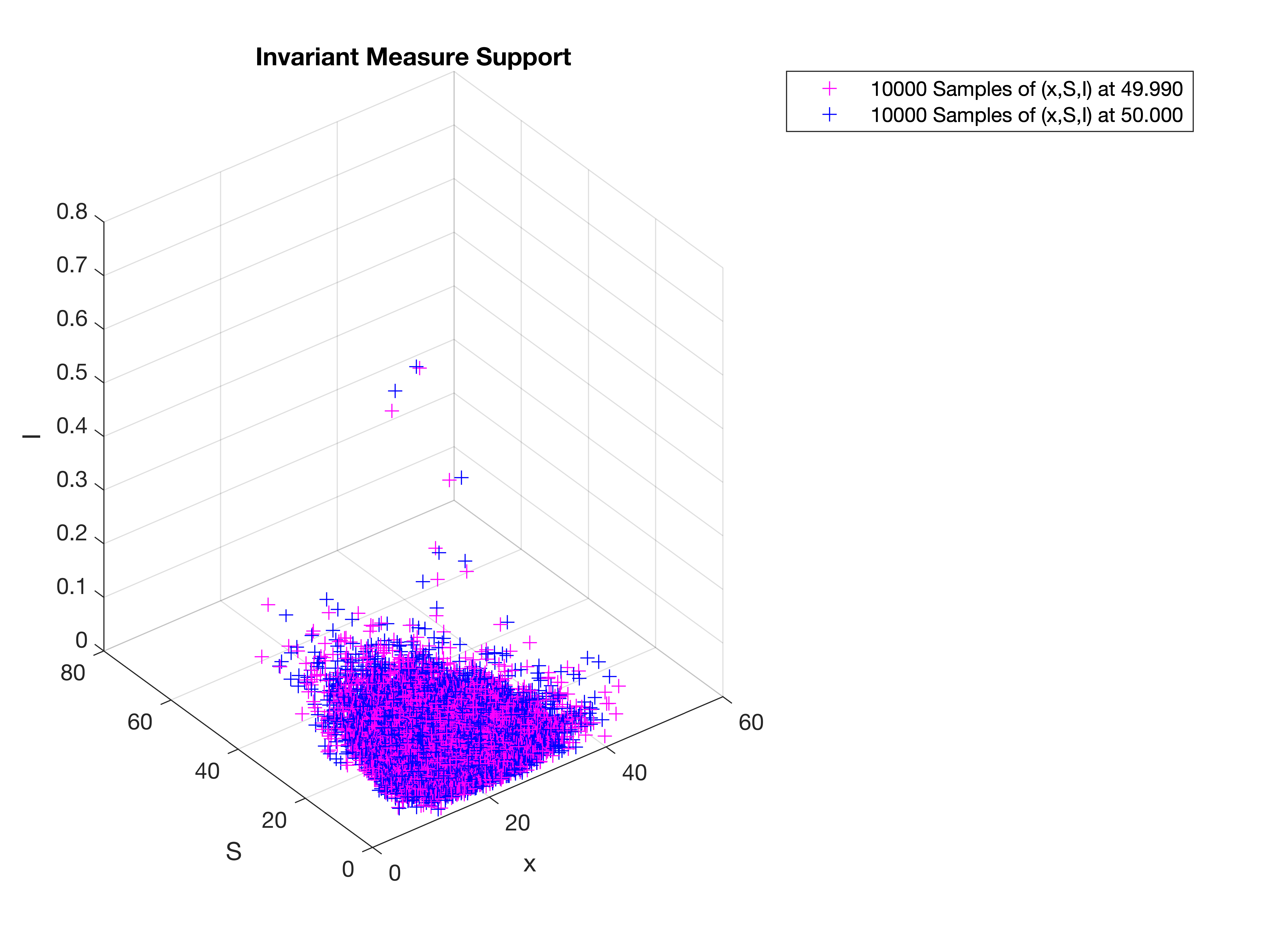}
    \caption{Distribution of $(x,S,I)$ with $a=0.6$, $b=2$, $a_1=2.7$, $a_2=1.2$, $b_1=0.1$, $c=1$, $h=0.01$, $\beta^*=1$, $k=0.1$, $\sigma_1=0.5$, $\sigma_2=0.6$, $\sigma_3=0.4$, $\gamma=0.3$, $\alpha=1.1$, and initial values $x(0)=6$, $S(0)=10$, and $I(0)=2$.}  \label{figure3}
  \end{center}
\end{figure}

\subsection{Decline}
We set the parameters as follows: $a=0.6$, $b=0.8$, $a_1=0.1$, $a_2=0.2$, $b_1=0.5$, $c=0.9$, $h=0.01$, $\beta^*=1$, $k=0.1$, $\sigma_1=0.1$, $\sigma_2=0.85$, $\sigma_3=0.95$, $\gamma=0.3$, $\alpha=0.4$. We take the initial values of $x$, $S$, and $I$ to be $0.4$, $0.5$, and $0.6$, respectively. We use $N=5000$ and $T=20$. These parameter values satisfy the conditions in Theorem \ref{thm:dos}.

Figure \ref{figure4} illustrates the decline of the predator species, with $S$ and $I$ both tending to $0$ as $t\to\infty$.

\begin{figure}[H]
  \begin{center}  
    \includegraphics[scale=0.7]{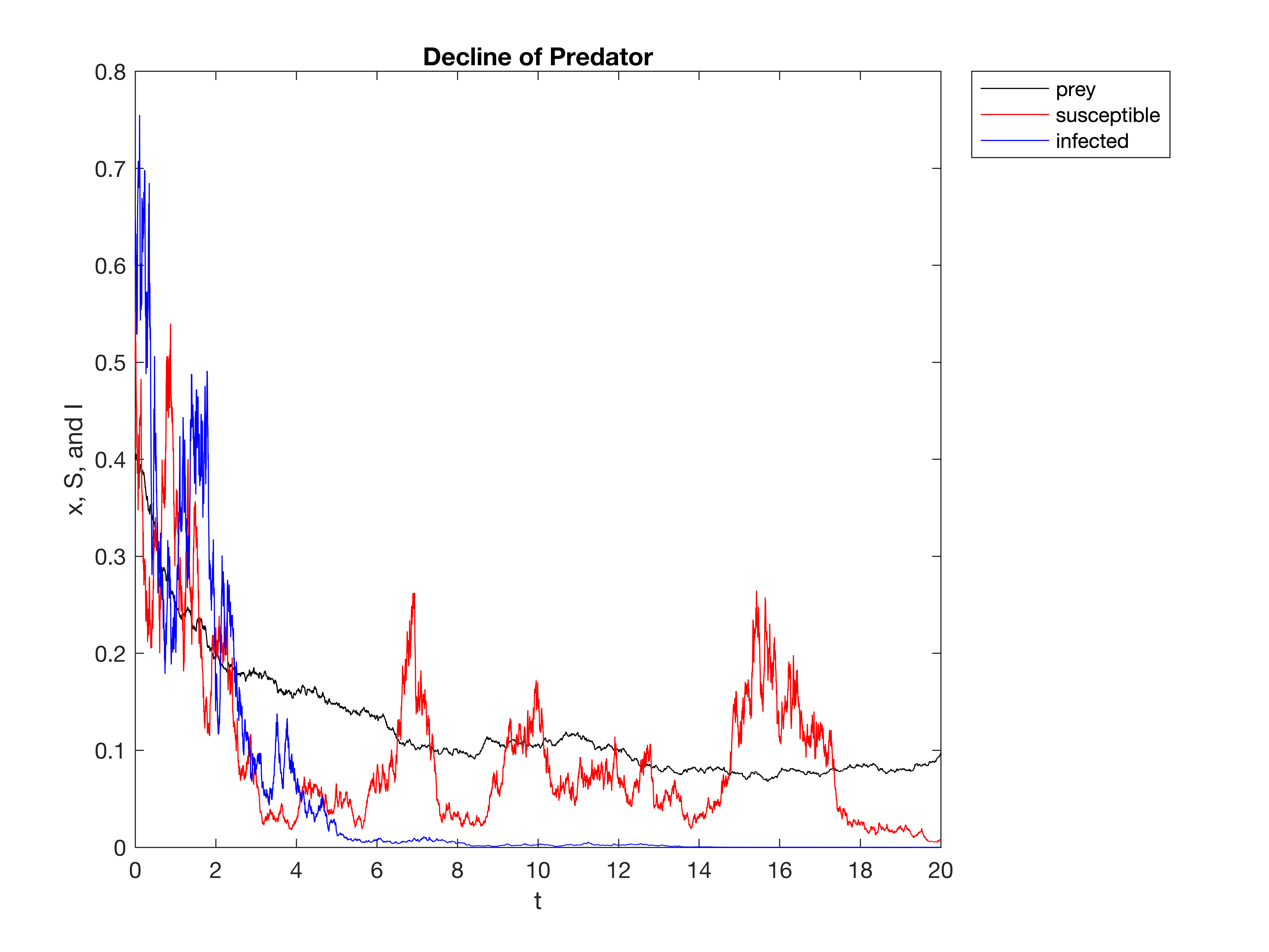}
    \caption{Decline of predator with $a=0.6$, $b=0.8$, $a_1=0.1$, $a_2=0.2$, $b_1=0.5$, $c=0.9$, $h=0.01$, $\beta^*=1$, $k=0.1$, $\sigma_1=0.1$, $\sigma_2=0.85$, $\sigma_3=0.95$, $\gamma=0.3$, $\alpha=0.8$, $x(0)=0.4$, $S(0)=0.5$, and $I(0)=0.6$.}  \label{figure4}
  \end{center}
\end{figure}

Next, we change the values of the noise intensity and $\alpha$  while keeping all other parameters the same as in Figure \ref{figure4}. The new parameter values are $\sigma_1=0.9, \sigma_2=0.2, \sigma_3=0.1, \alpha=0.4$ and they satisfy the conditions in Theorem \ref{thm:dop}. Figure \ref{figure5} illustrates the decline of the prey species, where $x$ tends to $0$.
\begin{figure}[H]
  \begin{center}  
    \includegraphics[scale=0.7]{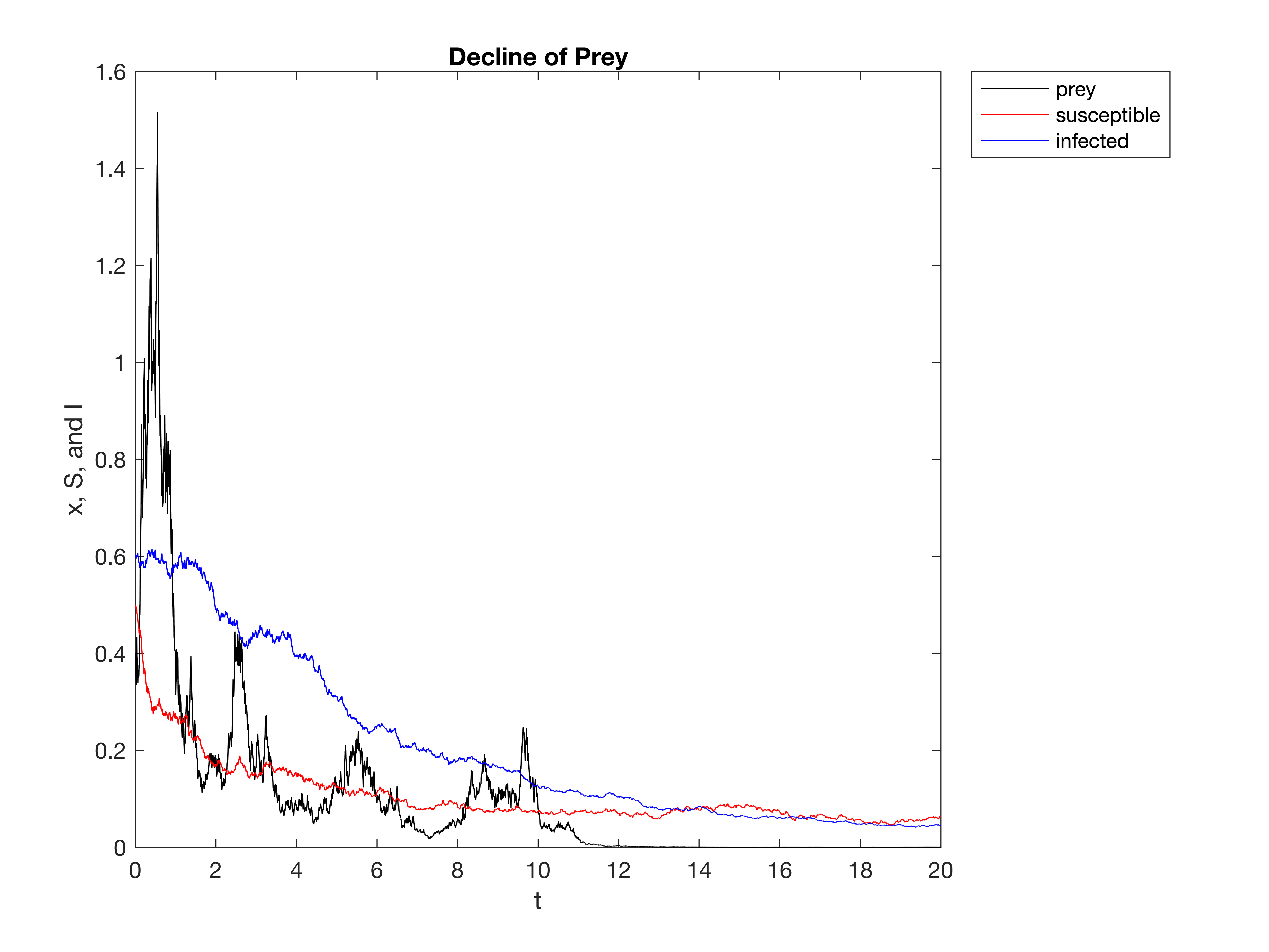}
    \caption{Decline of prey with $a=0.6$, $b=0.8$, $a_1=0.1$, $a_2=0.2$, $b_1=0.5$, $c=0.9$, $h=0.01$, $\beta^*=1$, $k=0.1$, $\sigma_1=0.9, \sigma_2=0.2, \sigma_3=0.1,$ $\gamma=0.3$, $\alpha=0.4$, $x(0)=0.4$, $S(0)=0.5$, and $I(0)=0.6$.}  \label{figure5}
  \end{center}
\end{figure}

\section{Conclusions}

In this paper, we have considered a prey-predator model with infected predator population and studied the dynamical behaviour under environmental driving forces. The predator population is assumed to be generalist in nature as they have an alternative implicit food source. The predator growth rate primarily follow the Leslie-Gower formulation but an modification is introduced to model the resource independent intra-specific competition rate explicitly. This modified formulation has an advantage - large density of favorable food source can not eliminate the strength of intra-specific competition. Mathematically, the modified specialist predator's growth rate helps to prove the global existence of solution for the stochastic differential equation model.  

Our study has provided analytical conditions for the sustainability of both prey and predator populations, as well as for the decline of each species. Additionally, we have explored the existence of a Borel invariant measure.

The analysis of sustainability conditions offers valuable insights into the long-term viability of prey and predator populations within their respective ecosystems. Understanding the factors that contribute to population sustainability is essential for effective conservation efforts and the preservation of balanced ecological systems.

Furthermore, our investigation of decline conditions for each species offers valuable information on the factors that can lead to population decreases or even local extinctions. Identifying these conditions helps identify potential threats or vulnerabilities in the ecosystem, enabling targeted conservation strategies to mitigate population decline and preserve biodiversity.

The discovery of a Borel invariant measure adds to our understanding of the long-term behavior of prey-predator populations and contributes to assessing the stability and resilience of the ecosystem.

In conclusion, our research holds significance for informing conservation strategies, ecosystem management, and enhancing our understanding of the intricate dynamics within ecological systems.

\section*{Acknowledgments}
The work of the first and last authors was partly supported by JSPS KAKENHI Grant Number 19K14555.
\bibliographystyle{plain}

\bibliography{reference}

\end{document}